\documentclass{article}
\pdfoutput=1
\usepackage{arxiv}

\usepackage[utf8]{inputenc} 
\usepackage[T1]{fontenc}    
\usepackage{hyperref}       
\usepackage{url}            
\usepackage{booktabs}       
\usepackage{amsfonts}       
\usepackage{nicefrac}       
\usepackage{microtype}      
\usepackage{lipsum}
\usepackage{siunitx}
\usepackage[linesnumbered,ruled,vlined]{algorithm2e}
\usepackage{float}
\usepackage{amsmath}
\usepackage{amsthm}
\usepackage{mathtools}
\usepackage{enumerate}
\usepackage{cite}
\usepackage{xcolor}
\usepackage{graphicx}
\usepackage{braket}
\usepackage{comment}
\usepackage{caption}
\usepackage{cleveref}
\usepackage{subcaption}
\graphicspath{ {./images/} }
\newcommand{\tr}{\mathrm{tr}}

\newtheorem{proposition}{Proposition}

\newtheorem{definition}{Definition}
\newtheorem{theorem}{Theorem}
\newtheorem{lemma}{Lemma}
\DeclareMathOperator*{\argmin}{argmin}

\title{Efficient Hamiltonian-aware Quantum Natural Gradient Descent for Variational Quantum Eigensolvers}

\author{
  Chenyu Shi\\
  Applied Quantum Algorithms Leiden, Leiden University, Leiden, The Netherlands\\
  Leiden Institute of Advanced Computer Science, Leiden University, Leiden, The Netherlands\\
  \texttt{c.shi@liacs.leidenuniv.nl}\\
  \And
  Hao Wang\thanks{Corresponding author}\\
  Applied Quantum Algorithms Leiden, Leiden University, Leiden, The Netherlands\\
  Leiden Institute of Advanced Computer Science, Leiden University, Leiden, The Netherlands\\
  \texttt{h.wang@liacs.leidenuniv.nl}\\
}

\begin{document}
\maketitle

\begin{abstract}
The Variational Quantum Eigensolver (VQE) is one of the most promising algorithms for current quantum devices. It employs a classical optimizer to iteratively update the parameters of a variational quantum circuit in order to search for the ground state of a given Hamiltonian. The efficacy of VQEs largely depends on the optimizer employed. Recent studies suggest that Quantum Natural Gradient Descent (QNG) can achieve faster convergence speed than vanilla gradient descent (VG), but at the cost of additional quantum resources to estimate the Fubini-Study metric tensor in each optimization step. The Fubini-Study metric tensor used in QNG is related to the entire quantum state space and does not incorporate information about the target Hamiltonian. To take advantage of the structure of the Hamiltonian and address the limitation of additional computational cost in QNG, we propose Hamiltonian-aware Quantum Natural Gradient Descent (H-QNG). In H-QNG, we propose to use the Riemannian pullback metric induced from the lower-dimensional subspace spanned by the Hamiltonian terms onto the parameter space. We show that H-QNG inherits the desirable features of both approaches: the low quantum computational cost of VG and the reparameterization-invariance of QNG. We also validate its performance through numerical experiments on molecular Hamiltonians, showing that H-QNG achieves faster convergence to lower energy while requiring fewer quantum computational resources.
\end{abstract}

\section{Introduction}
The Variational Quantum Eigensolver (VQE) \cite{peruzzo2014variational} is a hybrid quantum–classical algorithm well-suited for current noisy intermediate-scale quantum devices \cite{preskill2018quantum}. It is a variational quantum algorithm \cite{cerezo2021variational} specifically designed to search for the ground state of a given Hamiltonian. VQEs have attracted wide attention across different research fields, such as quantum chemistry \cite{lim2024fragment, de2024evaluating}, quantum many-body problems \cite{gyawali2022adaptive, weaving2025simulating}, and materials engineering \cite{gujarati2023quantum, huang2022simulating}, and are believed to hold potential for quantum advantage over classical methods \cite{peruzzo2014variational, kandala2017hardware}. The basic idea of VQEs is to employ a classical optimizer to iteratively update the parameters of a parametrized quantum circuit in order to minimize the expectation value of the Hamiltonian, where the expectation value is evaluated via the output of the parameterized quantum circuit. The optimal parameters obtained by this minimization correspond to an approximate ground state of the given Hamiltonian. 

Despite its promising prospects across different fields, one of the main limitations of VQEs lies in their trainability. Recent studies have shown that VQEs can suffer from serious training issues \cite{mcclean2018barren, larocca2025barren}. Finding a suitable optimizer is believed to be crucial for revealing the potential quantum advantages of VQEs. The basic optimization method is gradient descent, which we refer to as vanilla gradient descent in this paper. Many improved gradient-based methods have been proposed for VQEs \cite{sweke2020stochastic, jones2025benchmarking, stokes2020quantum}. Among these methods, Quantum Natural Gradient Descent (QNG) has emerged as a promising approach \cite{stokes2020quantum} as the quantum analog of classical natural gradient descent \cite{amari1998natural}. QNG leverages the Fubini-Study metric tensor to modify the gradient in each optimization step to capture the geometric information of quantum states. 

QNG is believed to offer several advantages over vanilla gradient descent. First, like its classical counterpart, QNG is reparameterization-invariant \cite{stokes2020quantum, martens2020new}, whereas vanilla gradient descent generally is not. Reparameterization invariance is considered a desirable property of optimization methods, as it enhances the efficiency of the optimization process \cite{song2018accelerating}. In addition, QNG is potentially helpful for avoiding local minima—one of the most significant issues affecting the trainability of VQEs \cite{wierichs2020avoiding}. Moreover, several experimental studies have shown that QNG achieves faster convergence speed than vanilla gradient descent \cite{stokes2020quantum, koczor2022quantum, yamamoto2019natural}.

However, despite its advantages, QNG requires additional quantum computational cost compared to vanilla gradient descent. In each optimization step of QNG, the Fubini-Study metric must be experimentally estimated on hardware \cite{koczor2022quantum, haug2022natural}. For a general parameterized circuit in VQEs with $m$ parameters and a target Hamiltonian $H=\sum_{r=1}^v a_r P_r$ consisting of $v$ Pauli terms, the quantum computational cost is $\operatorname{O}(mv+m^2)$ in terms of circuit evaluation complexity for each optimization step. The $mv$ term arises from computing the gradient with respect to each parameter, while the $m^2$ term corresponds to evaluating the metric tensor. By contrast, each vanilla gradient descent step only requires $\operatorname{O}(mv)$ quantum computational cost for the gradient. Moreover, estimating the Fubini-Study metric tensor also requires additional circuits for the Hadamard test \cite{koczor2022quantum, koczor2022quantum2}. This raises an interesting question: Is it possible to construct an optimization method that requires the same quantum computational cost as VG, yet achieves performance comparable to or even better than QNG?

The inspiration for our method comes from the observation that, in the general case, an $n$-qubit target Hamiltonian $H=\sum_{r=1}^v a_r P_r$ does not contain all $4^n$ Pauli strings and is therefore associated only with a non-trivial lower-dimensional subspace spanned by each term $P_r$. For example, a one-qubit Hamiltonian $H=-X-Y$ is independent of the $z$-axis, and its expectation value depends only on the coordinates in the $xy$-plane. Hence, $H$ is related only to the unit disc in the $xy$-plane, rather than to the entire Bloch sphere. An illustration is provided in \Cref{fig-submanifold}.

Motivated by this observation, we propose considering the metric tensor used in optimization directly induced from this non-trivial lower-dimensional space. Accordingly, our method employs the Riemannian pullback metric induced from this subspace spanned by Hamiltonian terms onto the parameter space (see \Cref{df-hqng}), instead of using the Fubini-Study metric tensor defined over the entire quantum state space. Since the metric tensor in our method incorporates information from the target Hamiltonian, we refer to our approach as \textit{Hamiltonian-aware Quantum Natural Gradient Descent} (H-QNG).

In this work, we show that H-QNG requires the same quantum computational cost $\operatorname{O}(mv)$ as vanilla gradient descent. Moreover, we prove that H-QNG preserves the reparameterization-invariant property of standard QNG. Besides, by incorporating the structure of the Hamiltonian into the metric tensor, H-QNG can potentially achieve faster convergence compared to standard QNG. We experimentally display this improved convergence speed performance in this work. All experiments in the main text of this work are performed on the classical simulator, where exact quantities can be tracked noiselessly. Particularly, we also discuss the impact of shot noise through additional experiments in Appendix~\ref{NOISE}.

The remainder of the paper is structured as follows. \Cref{Pre} briefly introduces the background knowledge required for this work. \Cref{MF} formulates H-QNG and reveals its differences and connections to standard QNG. \Cref{MA} analyzes the key properties of H-QNG and provides experimental application examples to display its performance. Finally, \Cref{DC} concludes the paper.

\section{Preliminaries} \label{Pre}

\subsection{Variational Quantum Eigensolver}
The Variational Quantum Eigensolver (VQE) \cite{peruzzo2014variational} is a hybrid quantum–classical algorithm designed to search for the ground state $\rho_G$ of a given target Hamiltonian $H$. The general workflow of VQEs is shown in \Cref{fig-vqe}. In VQEs, a parameterized quantum circuit is used to prepare a variational state $\rho_{\theta}=U(\theta)\ket{0}\bra{0}U^{\dag}(\theta)$ where $\theta \in \mathbb{R}^m$. Given a target Hamiltonian $H$, the expectation value $f(\theta)=\tr(\rho_{\theta}H)$ is estimated on quantum devices. The ground state $\rho_G$ is the lowest-energy state of the target Hamiltonian $H$. Hence, searching for $\rho_G$ is equivalent to minimizing the expectation value function $f(\theta)$ by adjusting the tunable parameters $\theta$.

Considering the expectation value $f(\theta)$ as the cost function, we can apply a classical optimizer to iteratively update the parameters $\theta$, while evaluating the required quantities on the quantum device. The most common yet simple method is gradient descent (which we refer to as vanilla gradient descent in this paper to distinguish it from other gradient-based methods). In vanilla gradient descent, the parameter update rule is given by:    

\begin{equation}
    \theta^{(k+1)}=\theta^{(k)}-\eta\nabla f(\theta^{(k)})
    \label{eq-vanilla}
\end{equation}

where $\nabla f(\theta^{(k)})$ denotes the gradient of the cost function (expectation value), and $\eta$ is the learning rate. In each optimization step, the gradients are evaluated on quantum devices using techniques such as finite differences or the parameter-shift rule \cite{mitarai2018quantum,schuld2019evaluating}. Beyond vanilla gradient descent, many gradient-based methods have been proposed and evaluated for VQEs \cite{jones2025benchmarking}. With a good optimizer and an expressive circuit, VQEs are expected to obtain a good approximation of the ground state $\rho_G$ after sufficient optimization.

\begin{figure}[h]
    \centering
    \includegraphics[width=0.6\textwidth]{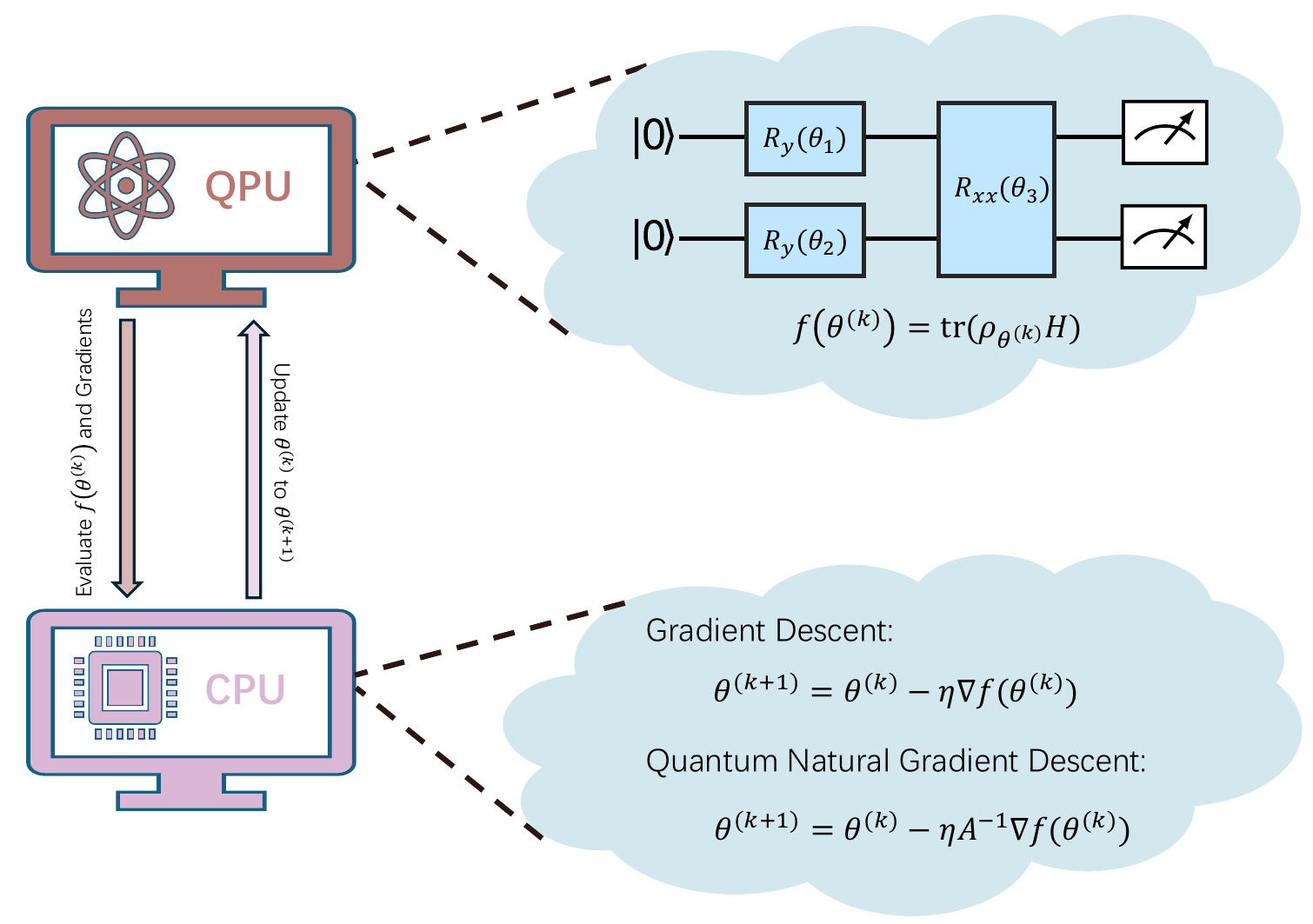}  
    \caption{The general workflow of VQEs. In the quantum part, a parameterized quantum circuit is used to estimate the expectation value $\tr(\rho_{\theta^{(k)}}H)$ and its gradient, given the current parameters $\theta^{(k)}$. In the classical part, a classical optimizer iteratively updates the parameters from $\theta^{(k)}$ to $\theta^{(k+1)}$ in order to minimize the expectation value. Various optimizers can be employed in the optimization of VQEs, and the choice plays an important role in its performance.}
    \label{fig-vqe}
\end{figure}

\subsection{Quantum Natural Gradient Descent}

In each optimization step, vanilla gradient descent solves the following constrained optimization problem \cite{shi2025weighted}:

\begin{equation}\label{vanilla_problem}
\begin{array}{cl}
\min\limits_{\delta} &\; f(\theta + \delta) \\
\mbox{s.t.} &\; \Vert \delta \Vert^2_2 \leq \varepsilon
\end{array}
\end{equation}

This constrained optimization problem minimizes the cost function within a local neighborhood of the current parameters $\theta^{(k)}$. In Appendix~\ref{cons}, we show that the constrained optimization problem above yields the vanilla gradient descent update formula \Cref{eq-vanilla} via a first-order Taylor approximation. One limitation of vanilla gradient descent is that each optimization step is inherently tied to the Euclidean geometry of the parameter space, since the Euclidean distance is used to define the local neighborhood in \Cref{vanilla_problem}. It has been shown that the Euclidean geometry is suboptimal for optimizing quantum variational algorithms \cite{stokes2020quantum, harrow2021low}, and the optimization outcome may vary under different reparameterizations \cite{gacon2021simultaneous, stokes2020quantum}. This can lead to potential inefficiencies in the optimization of VQEs \cite{stokes2020quantum, amari1998natural}.

Quantum natural gradient descent (QNG) \cite{stokes2020quantum} is a gradient-based optimization method designed for variational quantum algorithms. It is the quantum analog of natural gradient descent \cite{amari1998natural} in classical machine learning. Instead of only considering Euclidean geometry in the parameter space, QNG accounts for the quantum geometry of the state space to modify the gradient. In each optimization step, QNG solves the following constrained optimization problem, where the local neighborhood is defined by fidelity distance between quantum states:

\begin{equation}\label{qng_problem}
\begin{array}{cl}
\min\limits_{\delta} &\; f(\theta + \delta) \\
\mbox{s.t.} &\; D_{F}(\rho_{\theta}, \rho_{\theta +\delta}) \leq \varepsilon
\end{array}
\end{equation}

where the fidelity distance $D_F(\rho,\sigma)=1 - (\tr(\sqrt{\sqrt{\rho}\sigma \sqrt{\rho}}))^2$. Similarly, \Cref{qng_problem} can also yield the update rule of QNG. In \Cref{df-qng}, we present the formal definition of QNG and its update rule.

\begin{definition}
    (Quantum Natural Gradient Descent). For pure variational quantum states $\ket{\phi_{\theta}}$ with $\theta=(\theta_1,\cdots,\theta_m)^T$, the quantum natural gradient descent step with respect to the expectation value $f(\theta)=\braket{\phi_{\theta}|H|\phi_{\theta}}$ is defined as:
    \begin{equation}
        \theta^{(k+1)} = \theta^{(k)} - \eta A^{-1}\nabla f(\theta^{(k)})
        \label{eq-qng}
    \end{equation}

    where $A_{ij}=\operatorname{Re}[\braket{\partial_i \phi_{\theta}|\partial_j\phi_{\theta}}-\braket{\partial_i\phi_{\theta}|\phi_{\theta}}\braket{\phi_{\theta}|\partial_j \phi_{\theta}}]$ is the Fubini-Study metric tensor.
\label{df-qng}
\end{definition}

The derivation from \Cref{qng_problem} to \Cref{eq-qng} is shown in Appendix~\ref{cons}. As we can see, the Fubini-Study metric tensor is used to modify the gradient in each optimization step. To address the limitation of vanilla gradient descent, QNG takes into account the geometry of the quantum state, and it is reparameterization-invariant \cite{stokes2020quantum}. Several experimental studies further indicate that QNG achieves faster convergence speed than vanilla gradient descent \cite{stokes2020quantum, koczor2022quantum, yamamoto2019natural}.

Despite its many advantages, a main limitation of QNG is the requirement of additional quantum computational costs. As shown in \Cref{eq-qng}, the Fubini-Study metric tensor is required in each optimization step. In the general case, this metric tensor is not available in analytic form and must be estimated empirically on hardware \cite{ferreira2025variational}. To address this limitation, we propose a method that avoids additional quantum computational costs.

\section{Method Formulation} \label{MF}
In this section, we formally present the proposed method in \Cref{df-hqng}. Before introducing our definition, we first review the notion of the pullback metric \cite{augenstein2024probabilistic, bai2022geometric, diepeveen2024score} and explain how it can be used to perform gradient-based optimization. We also show that QNG, as defined in \Cref{df-qng}, can be interpreted from the perspective of the pullback metric in \Cref{pp-qng}. Detailed proofs for the results in this section are provided in Appendix~\ref{pf-1}. We also reveal the connections and differences between our method and the related work OP-VQITE \cite{anuar2024operator} in Appendix~\ref{COPE}.

\begin{definition}
    (Pullback Metric). Consider an immersion $h:\mathcal{S}\rightarrow \mathcal{M}$ from a parameter space $\mathcal{S}$ to a Riemannian Manifold $\mathcal{M}$ equipped with a Riemannian metric $G_{\mathcal{M}}$. A pullback metric $G_{\theta}$ on $\mathcal{S}$ is induced by the immersion $h$. For $\theta\in\mathcal{S}$ and $u,v \in \mathcal{T_{\theta}S}$, $G_{\theta}$ is given by:
    \begin{equation}
        G_{\theta}(u,v) = G_{\mathcal{M}}(dh_{\theta}(u), dh_{\theta}(v))
    \end{equation}

    Particularly, if $\mathcal{S}$ and $\mathcal{M}$ are both Euclidean spaces with coordinates, then the matrix form of the pullback metric is given by:
    \begin{equation}
        G=J_h^TJ_h
    \end{equation}    
    where $J_h$ is the Jacobian of $h$ at $\theta$. Hence, the element of the matrix $ G_{ij}=\langle\partial_ih(\theta),\partial_jh(\theta)\rangle$.
\label{df-pullback}
\end{definition}

A detailed discussion of the above definition of the pullback metric can be found in \cite{augenstein2024probabilistic}. The Riemannian pullback metric can be used to define a gradient-based optimization \cite{bai2022geometric} with respect to a cost function $f=g\circ h:\mathcal{S}\rightarrow \mathbb{R}$ where $h:\mathcal{S}\rightarrow \mathcal{M}$ and $g:\mathcal{M}\rightarrow \mathbb{R}$. For each $\theta \in \mathcal{S}$, the gradient descent with the pullback metric is defined as: 
\begin{equation}
    \theta^{(k+1)} = \theta^{(k)} - \eta G^{-1}\nabla f(\theta^{(k)}) 
    \label{eq-rmg}
\end{equation}

where $G$ is the matrix form of the pullback metric defined in \Cref{df-pullback}. Here, we show that the QNG defined in \Cref{df-qng} can be interpreted from the perspective of the gradient descent with Riemannian pullback metric in \Cref{pp-qng}.

\begin{proposition}
    For an $n$-qubit VQE, suppose parameter space $\mathcal{S}\subseteq \mathbb{R}^{m}$, and let $\mathcal{M} \subseteq \mathbb{R}^{4^n}$ be the image of the space of $2^n \times 2^n$ Hermitian matrices under the mapping: $Q\rightarrow[\tr(QP_1),\cdots,\tr(QP_{4^n})]$, where $Q$ is an arbitrary $2^n\times 2^n$ Hermitian matrix and $P_r \ (1 \leq r \leq 4^n)$ are the possible Pauli strings. The metric $G_{\mathcal{M}}$ is the Euclidean Riemannian metric in $\mathbb{R}^{4^n}$. The gradient descent with Riemannian pullback metric for the cost function $f(\theta)=\tr(\rho_{\theta}H)$ can be formulated as \Cref{eq-rmg}, where the pullback metric $G$ is equal to Fubini-Study metric tensor $A$ up to a factor: $G=2^{n+1}A$. 
    \label{pp-qng}
\end{proposition}

\Cref{pp-qng} shows that the Riemannian pullback metric differs from the Fubini-Study metric tensor in QNG only by a constant factor. Since the constant factor $c$ can be absorbed into the learning rate, QNG defined in \Cref{df-qng} is fully equivalent to the gradient descent with Riemannian pullback metric in \Cref{pp-qng}.

\begin{figure}[htbp]
  \centering
  \begin{subfigure}[b]{0.4\textwidth}
    \includegraphics[width=\textwidth]{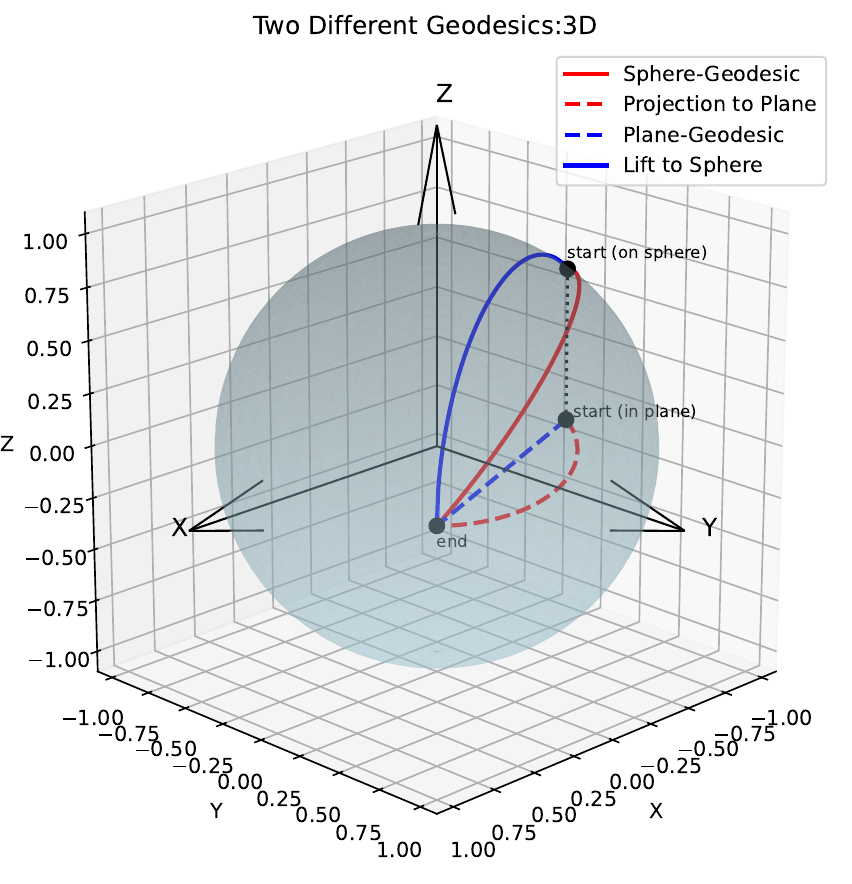}
    \caption{}
  \end{subfigure}
  \hspace{0.05\textwidth}
  \begin{subfigure}[b]{0.4\textwidth}
    \includegraphics[width=\textwidth]{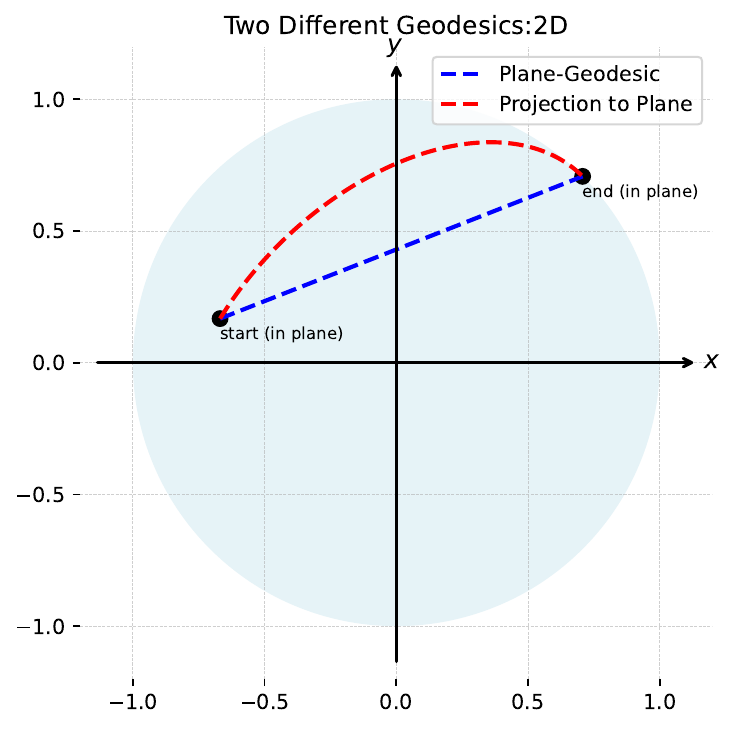}
    \caption{}
  \end{subfigure}
  \caption{To minimize the expectation value of the Hamiltonian $H=-X-Y$, only the manifold embedded in lower dimensional space (the unit disc with boundary in (b)) contributes non-trivially to the optimization rather than the entire Bloch sphere (the sphere in (a)), since the $z$-axis coordinate does not appear in the cost function. The geodesic with respect to the spherical metric (red solid line in (a)) represents the shortest path on the Bloch sphere from the initial point to the solution. However, its projection onto the non-trivial $xy$-plane (red dashed lines in a and b) is not the shortest path compared to the true geodesic within the $xy$-plane (blue dashed lines in a and b). Therefore, in the pullback metric view, H-QNG directly employs $g_{\mathcal{M}}$ as the metric of the non-trivial manifold in \Cref{df-hqng}, in contrast to standard QNG which considers $G_{\mathcal{M}}$ to be the metric of the full quantum state space in \Cref{pp-qng}.}
  \label{fig-submanifold}
\end{figure}

The function $h:\mathcal{S}\rightarrow \mathcal{M}$ in \Cref{pp-qng} maps a parameter $\theta \in \mathbb{R}^m$ to the vector $[\tr(\rho_{\theta}P_1),\cdots \tr(\rho_{\theta}P_{4^n})]\in \mathbb{R}^{4^{n}}$, and the corresponding function $g:\mathcal{M}\rightarrow \mathbb{R}$ is a linear combination of these $4^n$ elements such that $g\circ h(\theta)=f(\theta)=\tr(\rho_\theta H)$. However, the target Hamiltonian $H=\sum_{r=1}^{v}a_rP_r$ typically involves only a subset of the $4^n$ Pauli operators in general. Therefore, we can restrict the function $h$ to include only the relevant $v$ terms and discard the others. In other words, the cost function $f(\theta)$ is only determined by a manifold embedded in the lower dimensional space $\mathcal{M} \subseteq \mathbb{R}^v$. A one-qubit example is shown in Figure \ref{fig-submanifold}. The manifold $\mathcal{M}$ here represents the Bloch sphere. However, for the Hamiltonian $H=-X-Y$, only the $2$-dimensional manifold (a unit disc with boundary) in $XY$-plane contributes to the optimization since the objective function does not contain any component in the $z$-axis direction. Moreover, since the coefficients of the linear combination are determined by the Hamiltonian, they can be directly incorporated into the function $h$ as well, which rescales the lower-dimensional manifold. Based on this observation, we define a new gradient-based method with Riemannian pullback metric directly on the non-trivial manifold as the proposed method of this paper in \Cref{df-hqng}. We name it as \textit{Hamiltonian-aware Quantum Natural Gradient Descent} (H-QNG), to distinguish it from the standard QNG, where the Fubini-Study metric is not related to the given Hamiltonian.

\begin{definition}
    (Hamiltonian-aware Quantum Natural Gradient Descent (H-QNG)). For an $n$-qubit VQE with Hamiltonian $H=\sum_{r=1}^{v}a_rP_r$, suppose parameter space $\mathcal{S}\subseteq \mathbb{R}^{m}$, and let $\mathcal{M} \subseteq \mathbb{R}^{v}$ be the image of the space of $2^n \times 2^n$ Hermitian matrices under the mapping: $Q\rightarrow[a_1\tr(QP_1),\cdots,a_v\tr(QP_{v})]$, where $Q$ is an arbitrary $2^n\times 2^n$ Hermitian matrix and $P_r \ (1 \leq r \leq v)$ are the Pauli strings in the Hamiltonian $H$. The metric $g_{\mathcal{M}}$ is the Euclidean Riemannian metric in $\mathbb{R}^{v}$. The gradient descent with Riemannian pullback metric for the cost function $f(\theta)=\tr(\rho_{\theta}H)$ can be formulated as \Cref{eq-rmg}, where the pullback metric $G_{ij} = \sum_{r=1}^v a_r^2\tr(\partial_i\rho_{\theta}P_r)\tr(\partial_j\rho_{\theta}P_r)$. The Hamiltonian-aware Quantum Natural Gradient Descent is defined as:
    \begin{equation}
        \theta^{(k+1)} = \theta^{(k)} - \eta T^{-1}\nabla f(\theta^{(k)})
        \label{eq-hqng}
    \end{equation}
    where $T_{ij}=\frac{1}{2\sqrt{\sum_{r=1}^va_r^2}}G_{ij}$. In the case where $v=4^n$ and each $a_r=1$, H-QNG reduces to standard QNG.
    \label{df-hqng}
\end{definition}

The prefactor $\frac{1}{2\sqrt{\sum_{r=1}^va_r^2}}$ is introduced to ensure that the definition of H-QNG remains consistent with standard QNG: when $v=4^n$ and each $a_r=1$, the manifold in \Cref{df-hqng} reduces to the manifold in \Cref{pp-qng}, in which case we expect that H-QNG reduces to standard QNG. To realize this consistency, the prefactor must equal $2^{-(n+1)}$ in this condition according to \Cref{pp-qng}. Hence, the prefactor $\frac{1}{2\sqrt{\sum_{r=1}^va_r^2}}$ is used to appropriately scale $G_{ij}$.

As shown in \Cref{df-hqng}, the metric tensor $G$ does not require additional quantum measurement cost, as each gradient term in its expression is already required for evaluating $\nabla f(\theta)$. In contrast, the standard QNG demands extra quantum resources to compute the Fubini-Study metric tensor. The detailed analysis about computational cost scaling will be discussed in \Cref{SAS}. 

One desired property of standard QNG is its invariance under diffeomorphism reparameterization \cite{stokes2020quantum, koczor2022quantum}. That is, the optimization of standard QNG remains unchanged under diffeomorphism transformations of the parameter space. This reparameterization invariance is believed to improve the optimization process, which is discussed in \cite{song2018accelerating}. In \Cref{RI}, we show and prove that H-QNG also possesses this property.

We note that the definition of the pullback metric in \Cref{df-pullback} requires the function $h$ to be an immersion. However, the map $h:\theta \in \mathbb{R}^m\rightarrow [a_1\tr(\rho_{\theta}P_1),\cdots a_v\tr(\rho_{\theta}P_v)] \in \mathbb{R}^v$ realized by a variational quantum circuit, is not always guaranteed to be an immersion. This may affect the condition number and the invertibility of the resulting matrix $G$ of H-QNG. In \Cref{SMT}, we show that the matrix $G$ is almost everywhere invertible in the typical application scenario of VQEs. We also introduce regularization methods to address cases where $G$ is non-invertible or ill-conditioned.

\section{Method Analysis} \label{MA}
\subsection{Scaling Analysis} \label{SAS}
In this subsection, we analyze the computational costs of H-QNG. We primarily focus on the quantum computational costs and compare the classical computational costs at the end of the subsection.

Firstly, we analyze the quantum measurement overhead of each optimization step, which is proportional to the number of quantities that are required to be estimated during each optimization step. Here we assume an equal number of measurement shots per quantity estimation. For the quantum measurement overhead of vanilla gradient descent and QNG, we present the following proposition.
\begin{proposition}
    For an $n$-qubit VQE with $m$ parameters and a given Hamiltonian $H=\sum_{r=1}^va_rP_r$ consisting of $v$ terms, the quantum measurement overheads of each optimization step of vanilla gradient descent and QNG are $\operatorname{O}(mv)$ and $\operatorname{O}(mv + m^2)$, respectively.
    \label{pp-qcvgqng}
\end{proposition}

\begin{proof}
    For vanilla gradient descent, each optimization step requires evaluating the gradient with respect to each parameter for each Hamiltonian term $\partial_{\theta_i}\tr(\rho_{\theta}P_r)$, using the parameter-shift rule \cite{wierichs2022general,sack2022avoiding}. Since each partial derivative requires estimating two quantities using the parameter-shift rule, the total measurement overhead per optimization step is proportional to $2mv$ \footnote{Here we assume that the Pauli terms in the Hamiltonian are measured separately, without exploiting commutativity. Throughout our analysis and approximation, we assume by default that the gates admit a two-term parameter-shift rule. Gates requiring more than a two-term rule affect only the constant prefactor.}. For QNG, in addition to the gradient, each optimization step requires estimating the Fubini-Study metric tensor, which is an $m\times m$ matrix. Because it is symmetric, there are approximately $\frac{m^2}{2}$ elements to be estimated. Each element can be evaluated using the Hadamard test \cite{koczor2022quantum, koczor2022quantum2}, which approximately involves four quantity estimations for each element via the parameter-shift rule. As a result, the total number of circuit evaluations per optimization step of QNG is $\operatorname{O}(mv+m^2)$.
\end{proof}

As shown in \Cref{pp-qcvgqng}, each optimization step of QNG requires higher quantum measurement overhead compared to vanilla gradient descent. For H-QNG, we now present the following theorem regarding its quantum computational complexity.

\begin{theorem}
    For an $n$-qubit VQE with $m$ parameters and a given Hamiltonian $H=\sum_{r=1}^va_rP_r$ consisting of $v$ terms, the quantum measurement overhead of each optimization step of H-QNG is $\operatorname{O}(mv)$.
    \label{th-qchqng}
\end{theorem}
\begin{proof}
    Note that computing each element of the matrix $T$ in \Cref{eq-hqng} only requires all the partial derivative terms, which have already been estimated during the gradient estimation. Therefore, estimating the matrix $T$ in H-QNG does not require any additional quantum measurement overhead, and the total cost remains $\operatorname{O}(mv)$.
\end{proof}

\Cref{pp-qcvgqng} and \Cref{th-qchqng} show that H-QNG requires the same quantum measurement overhead as vanilla gradient descent. Compared to standard QNG, each optimization step of H-QNG is more efficient in terms of quantum resource usage. Besides, H-QNG can even achieve a better convergence speed compared to standard QNG in terms of optimization steps. We conduct a numerical experiment to illustrate this result in \Cref{fig-scaling}.

\begin{figure}[h]
    \centering
    \includegraphics[width=0.45\textwidth]{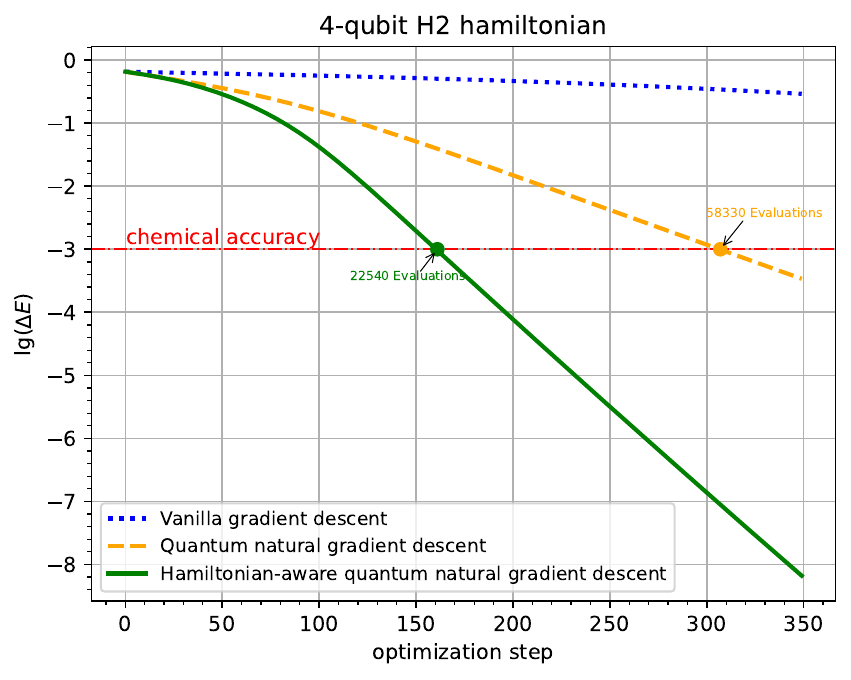}  
    \caption{The energy distance from the ground state $\Delta E$ of a 4-qubit H\textsubscript{2} Hamiltonian using vanilla gradient descent (blue dotted line), QNG (orange dashed line), and H-QNG (green solid line). All three methods start from the identical initial parameters. H-QNG reaches chemical accuracy (red dash-dotted line) earlier than standard QNG, both in terms of optimization steps and quantum circuit evaluations. Both standard QNG and H-QNG significantly outperform vanilla gradient descent. The annotated circuit-evaluation counts are approximate estimates based on the per-step quantity counts described in  \Cref{pp-qcvgqng} and \Cref{th-qchqng}.}
    \label{fig-scaling}
\end{figure}

\Cref{fig-scaling} shows the energy distance from the ground state for a 4-qubit H\textsubscript{2} Hamiltonian using vanilla gradient descent, standard QNG, and H-QNG. This experiment was implemented using PennyLane \cite{bergholm2018pennylane}. The Hamiltonian was obtained from the PennyLane Molecules dataset \cite{azad2023pennylaneqchem}. The variational ansatz is shown in Appendix~\ref{ansatz} for details. All three methods start from the identical initial parameters. In terms of optimization steps (horizontal axis), H-QNG (green solid line) reaches chemical accuracy (red dash-dotted line) faster than standard QNG (orange dashed line). We also report the approximate number of circuit evaluations required to achieve chemical accuracy in the unit of quantity to be estimated, using the per-step quantity counts described in \Cref{pp-qcvgqng} and \Cref{th-qchqng}. Since each optimization step of H-QNG has lower quantum measurement overhead, it reaches chemical accuracy even faster in terms of total circuit evaluations compared to standard QNG. Besides, we notice that vanilla gradient descent fails to achieve chemical accuracy within the given optimization budget. As H-QNG and vanilla gradient descent have the same measurement overhead, H-QNG significantly outperforms vanilla gradient descent in convergence speed under the same quantum computational budget.

In addition to measurement overhead, compared to standard QNG, H-QNG also requires lower circuit complexity. As discussed in \Cref{pp-qcvgqng}, estimating elements of the Fubini-Study metric tensor requires techniques such as the Hadamard test or swap test, which involve additional ancillary qubits or deeper quantum circuits. In contrast, H-QNG introduces no additional circuit complexity.

For classical computational costs, we still consider vanilla gradient descent as the baseline. We omit terms dominated by higher-order terms in the following estimates. Assuming that the required values have already been obtained from quantum measurements, vanilla gradient descent requires $\operatorname{O}(vm)$ classical operations in each step to accumulate the termwise gradient contributions in the parameter update. For QNG, the classical computational cost becomes $\operatorname{O}(vm+m^3)$ in each step, since computing the inverse of the metric tensor requires additional $\operatorname{O}(m^3)$ classical operations. For H-QNG, the classical cost becomes $\operatorname{O}(vm^2+m^3)$ in each step: the additional $\operatorname{O}(vm^2)$ operations are used to assemble the metric tensor from $v$ termwise gradient outer products; the same $\operatorname{O}(m^3)$ cost is required for metric-tensor inversion; and the $\operatorname{O}(vm)$ cost of accumulating the termwise gradient contributions and updating the parameters is absorbed into $\operatorname{O}(vm^2)$.

In conclusion, compared with QNG, H-QNG requires additional classical costs for assembling its metric tensor but requires fewer quantum circuit evaluations. Compared with vanilla gradient descent, H-QNG retains the same circuit-evaluation scaling while requiring additional classical cost for metric assembly and inversion. Specifically, H-QNG does not require an ancillary qubit to estimate the metric tensor compared with standard QNG.

\subsection{Reparameterization Invariance}
\label{RI}
One of the most important properties of standard QNG is the invariance under reparameterization given by a diffeomorphism\cite{stokes2020quantum, koczor2022quantum}. Here we give a brief definition to the reparameterization invariance property for an optimization method in \Cref{df-ri}.

\begin{definition}
    (Reparameterization Invariance for an Optimization Method). Consider a reparameterization given by a diffeomorphism $\theta=t(\psi)$. Suppose in a step of the optimization, the parameter $\theta$ is updated to $\theta'$ , and correspondingly $\psi$ is updated to $\psi'$. An optimization method is said to be reparameterization invariant if, whenever $\theta=t(\psi)$, it holds that $\theta'=t(\psi')$ after the update.
    \label{df-ri}
\end{definition}

To discuss the reparameterization invariance property of QNG and H-QNG, we express both methods in the form of differential equations, which we refer to as the continuous form of QNG and H-QNG. 
\begin{equation}
    \dot{\theta}(\tau)=-A_{\theta}^{-1}\nabla_{\theta} f(\theta)
    \label{eq-qngode}    
\end{equation}

\begin{equation}
    \dot{\theta}(\tau)=-T_{\theta}^{-1}\nabla_{\theta} f(\theta)
    \label{eq-hqngode}
\end{equation}

\Cref{eq-qngode} and \Cref{eq-hqngode} lead to \Cref{eq-qng} and \Cref{eq-hqng}, respectively, through the finite difference method \cite{zhou2012numerical}. In practical use, with a sufficiently small learning rate and an adequate number of optimization steps, \Cref{eq-qng} and \Cref{eq-hqng} can be regarded as good discrete approximations of \Cref{eq-qngode} and \Cref{eq-hqngode}, respectively. We specifically refer to \Cref{eq-qngode} and \Cref{eq-hqngode} as the continuous forms of QNG and H-QNG, to distinguish them from their usual discrete forms in \Cref{eq-qng} and \Cref{eq-hqng}. Based on these formulations, we further establish \Cref{pp-inva} and \Cref{th-inva}.

\begin{proposition}
    Continuous form of Quantum Natural Gradient Descent defined in \Cref{eq-qngode} is reparameterization invariant.
    \label{pp-inva}
\end{proposition}

The property in \Cref{pp-inva} has been discussed in \cite{stokes2020quantum,koczor2022quantum}. In \Cref{th-inva}, we demonstrate that H-QNG also possesses reparameterization invariance. The detailed proof of \Cref{th-inva} is provided in Appendix~\ref{pf-2}.

\begin{theorem}
    Continuous form of Hamiltonian-aware Quantum Natural Gradient Descent defined in \Cref{eq-hqngode} is reparameterization invariant.
    \label{th-inva}
\end{theorem}

\Cref{th-inva} and \Cref{df-ri} indicate that, just as in standard QNG, given any initial point $\theta{(0)}=t(\psi{(0)})$, H-QNG will follow two equivalent trajectories $\{\theta{(\tau)}\}$ and $\{\psi{(\tau)}\}$ in these two parameter spaces, where each $\theta{(\tau)}=t(\psi{(\tau)})$. Besides, note that \Cref{pp-inva} and \Cref{th-inva} assume that these gradient-based methods are in the continuous form, but in practice only a finite number of steps can be taken. With a small learning rate, QNG taking finite steps is approximately parameterization invariant \cite{martens2020new}. For H-QNG, a similar conclusion can be drawn as in \cite{martens2020new} for QNG. 

We also evaluate this property through numerical experiments where H-QNG can only take finite steps with small learning rates. The results are shown in \Cref{fig-ri}. Three different reparameterization functions $t$ are used. The optimization is performed in the $\psi$-space to obtain the trajectories $\{\psi^{(k)}\}$, with the same initial point in the $\theta$-space where $\theta^{(0)}=t(\psi^{(0)})$. The reparameterization invariance ensures that the trajectories $\{\theta^{(k)}\}$ in the $\theta$-space obtained by $t(\{\psi^{(k)}\})$ remain nearly unchanged regardless of different choices of $t$. As shown in the figures, the trajectories plotted in $\theta$-space remain nearly invariant under different reparameterizations for both QNG and H-QNG (the orange dashed curve and green solid curve). However, this invariance does not hold for vanilla gradient descent (the blue dotted curve). The ansatze used in this section are shown in Appendix~\ref{ansatz} in detail.

\begin{figure}[htbp]
  \centering
  \begin{subfigure}[b]{0.32\textwidth}
    \includegraphics[width=\textwidth]{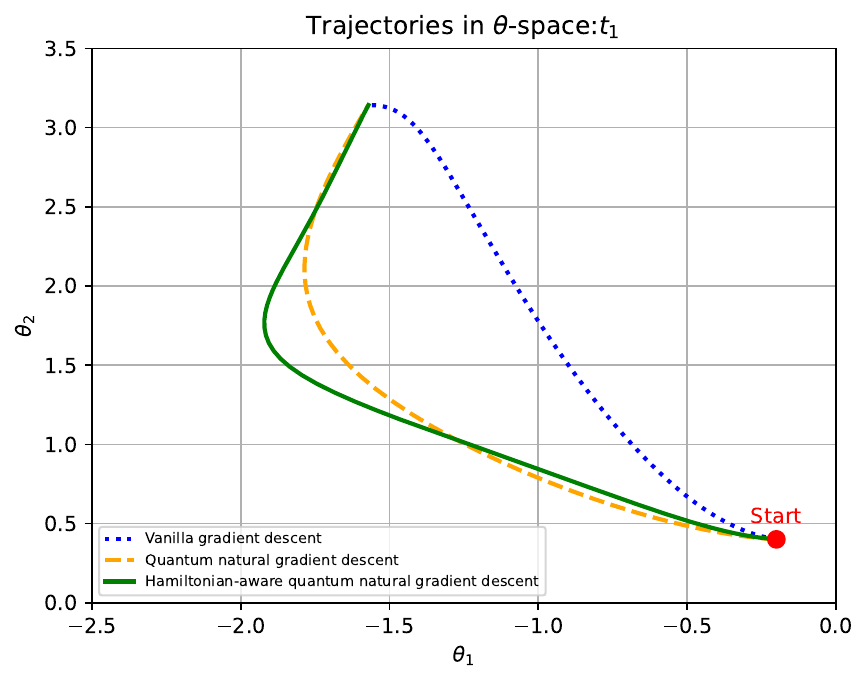}
    \caption{}
  \end{subfigure}
  \begin{subfigure}[b]{0.32\textwidth}
    \includegraphics[width=\textwidth]{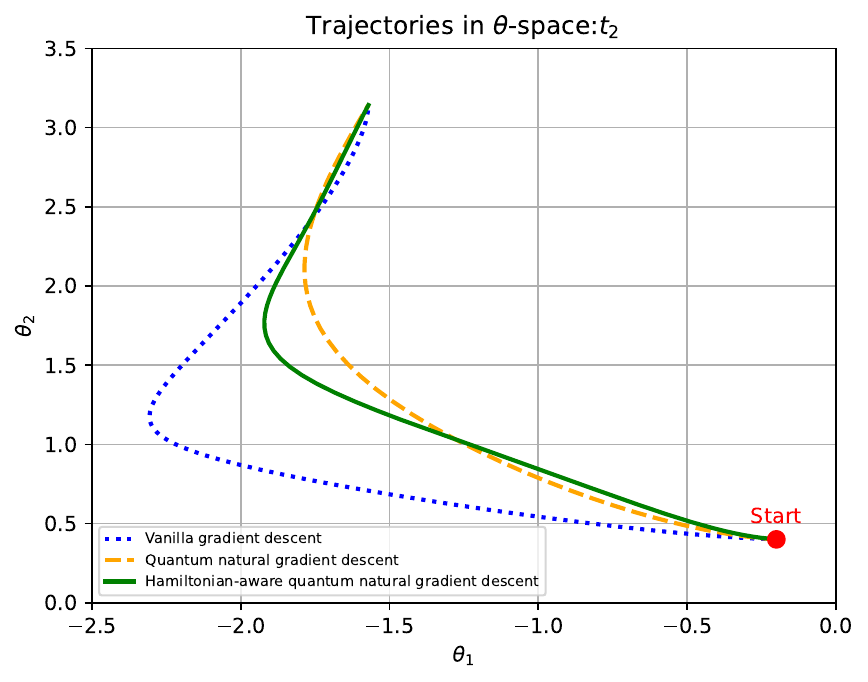}
    \caption{}
  \end{subfigure}
  \begin{subfigure}[b]{0.32\textwidth}
    \includegraphics[width=\textwidth]{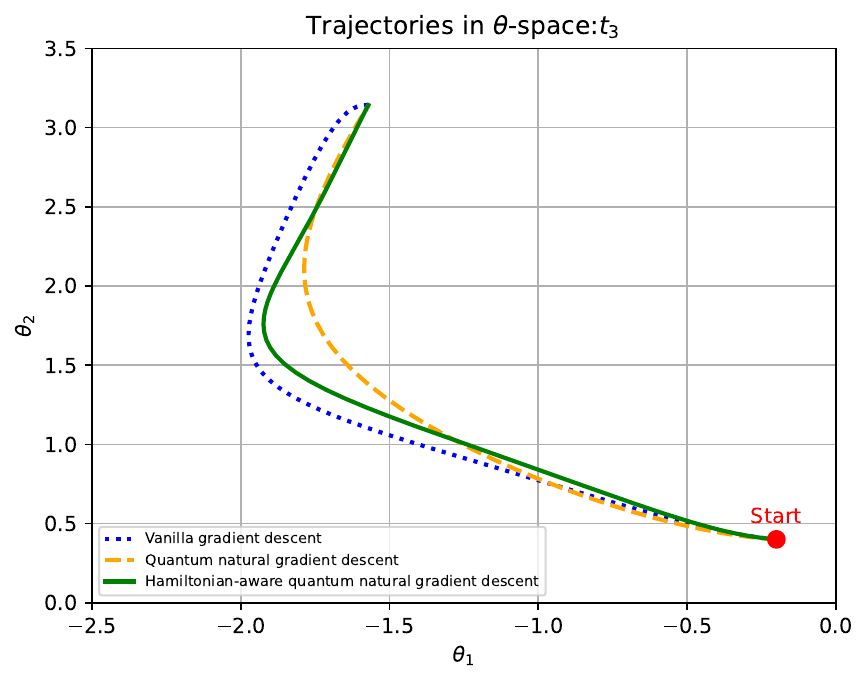}
    \caption{}
  \end{subfigure}
  \caption{Evaluating the reparameterization invariant property for vanilla gradient descent, QNG and H-QNG. The given $2$-qubit Hamiltonian is $H=X_1+Y_1+Z_1+X_2+Y_2+Z_2+X_1X_2+Y_1Y_2+Z_1Z_2$. The experiments are performed using a $2$-parameter VQE circuit composed of Pauli rotation and CNOT gates, evaluated under three optimization methods. The cost function is given by $\tr(\rho(\theta_1,\theta_2)H)$ , where the optimal solution can be analytically obtained as $\theta_1=-\frac{\pi}{2}$ and $\theta_2=\pi$. Optimization is carried out using three different reparameterization functions $t$ in the $\psi$-space where $t_1:(\theta_1,\theta_2)=(0.8\psi_1,1.2\psi_2)$, $t_2:(\theta_1,\theta_2)=(1.2\psi_1,0.8\psi_2)$, $t_3:(\theta_1, \theta_2)=(2\arctan(2\tan(0.5\psi_1)),2\arctan(2\tan(0.5\psi_2)))$. The initial points $\psi^{(0)}$ in $\psi$-space are chosen such that they correspond to the same point in $\theta$-space under each mapping $t$. The trajectories $\{\theta^{(k)}\}$ are obtained by applying each $t$ to the corresponding trajectories in the $\psi$-space. The trajectories of QNG and H-QNG remain nearly invariant in the $\theta$-space (the orange dashed curve and green solid curve), whereas this invariance does not hold for vanilla gradient descent (the blue dotted curve).}
  \label{fig-ri}
\end{figure}

\subsection{Singularity of the Metric Tensor}
\label{SMT}
\Cref{df-pullback} requires the function $h$ to be an immersion; otherwise, the pullback metric may be singular. However, this immersion condition does not always hold for variational quantum circuits. In this section, we prove that, in general, the metric tensor $T$ used in H-QNG is almost everywhere non-singular in the parameter space. We begin by formulating the following proposition for standard QNG. The detailed proofs of propositions and theorems in this subsection are shown in Appendix~\ref{pf-3}.
\begin{proposition}
    For an $n$-qubit VQE with $m$ parameters and a given Hamiltonian $H=\sum_{r=1}^va_rP_r$, the Fubini-Study metric tensor $A$ used in QNG can only be either:
    \begin{enumerate}
        \item Almost everywhere non-singular for $\theta$ in the parameter space.
        \item Everywhere singular for $\theta$ in the parameter space.
    \end{enumerate}
    \label{pp-sinqng}
\end{proposition}

In the general case of a well-constructed parameterization, the Fubini-Study metric tensor should not be constantly singular. Under this condition, \Cref{pp-sinqng} implies that singularities of the Fubini-Study metric tensor can only occur in a small subspace of measure zero within the entire parameter space. This property guarantees that the standard QNG remains well-defined at almost all points in the parameter space. We also aim for H-QNG to preserve this desirable property. Fortunately, as shown in the following theorem, this is indeed the case for H-QNG.

\begin{theorem}
   For an $n$-qubit VQE with $m$ parameters and a given Hamiltonian $H=\sum_{r=1}^va_rP_r$, the metric tensor $T$ used in H-QNG can only be either:
    \begin{enumerate}
        \item Almost everywhere non-singular for $\theta$ in the parameter space. 
        \item Everywhere singular for $\theta$ in the parameter space.
    \end{enumerate}
    \label{th-sinhqng} 
\end{theorem}

The specific singularity property that the metric tensor $T$ follows among the two possible ones in \Cref{th-sinhqng} depends on the relative magnitudes of $m$ and $v$, namely the number of parameters in the ansatz and the number of Pauli terms in the Hamiltonian. For the case $m>v$, the function $h$ fails to be an immersion, and the metric tensor $T$ is everywhere singular for $\theta$ in the parameter space. However, the condition $m>v$ is not typical for the application of VQEs on NISQ devices. NISQ quantum devices limit VQEs to shallow quantum circuits \cite{preskill2018quantum, lau2022nisq}, where the number of parameters can only grow moderately with the number of qubits $n$, such as $O(n\operatorname{polylog}n)$ \cite{koczor2022quantum}. In contrast, the number of terms in the Hamiltonian can grow much faster, such as $O(n^4)$ for molecular systems in quantum chemistry \cite{kojima2024orbital, koczor2022quantum}. This typically ensures that $m<v$ in typical applications of VQEs such as quantum chemistry.

In the more general case where $m<v$, the metric tensor $T$ used in H-QNG is non-singular almost everywhere in the parameter space under a good parameterization. This property implies that the metric tensor $T$ is well-defined throughout most of the parameter space, similarly to the case of standard QNG. 

However, \Cref{pp-sinqng} and \Cref{th-sinhqng} only ensure that the metric tensors are non-singular in most regions of the parameter space; they do not guarantee a good condition number. Moreover, the optimization trajectory may still pass through singular points, such as in the case of frustration-free Hamiltonians \cite{sattath2016local}, where the optimal point is singular. For instance, standard QNG uses the Moore–Penrose inverse of the Fubini-Study metric tensor instead of its exact inverse, as proposed in the original paper \cite{stokes2020quantum}. Alternatively, the regularization can also be adopted by adding a scaled identity matrix, i.e. $A+\lambda I$, where $\lambda$ is a tunable hyperparameter. For H-QNG, we adopt the latter method since it is more widely used in recent studies \cite{koczor2022quantum, motta2020determining, van2021measurement, shi2025weighted, halla2025quantum, bergholm2018pennylane}. The choice of $\lambda$ is crucial and requires careful tuning, as illustrated in \Cref{fig-lambda}. The other experimental settings are the same as those in \Cref{fig-scaling}. From the results, we can see that a value of $\lambda$ that is too large can lead to slow convergence, while a value that is too small may fail to adequately regularize the metric tensor, potentially causing oscillations, non-convergence, or early convergence. General strategies for selecting $\lambda$ are discussed in \cite{nocedal1999numerical}.

\begin{figure}[h]
    \centering
    \includegraphics[width=0.45\textwidth]{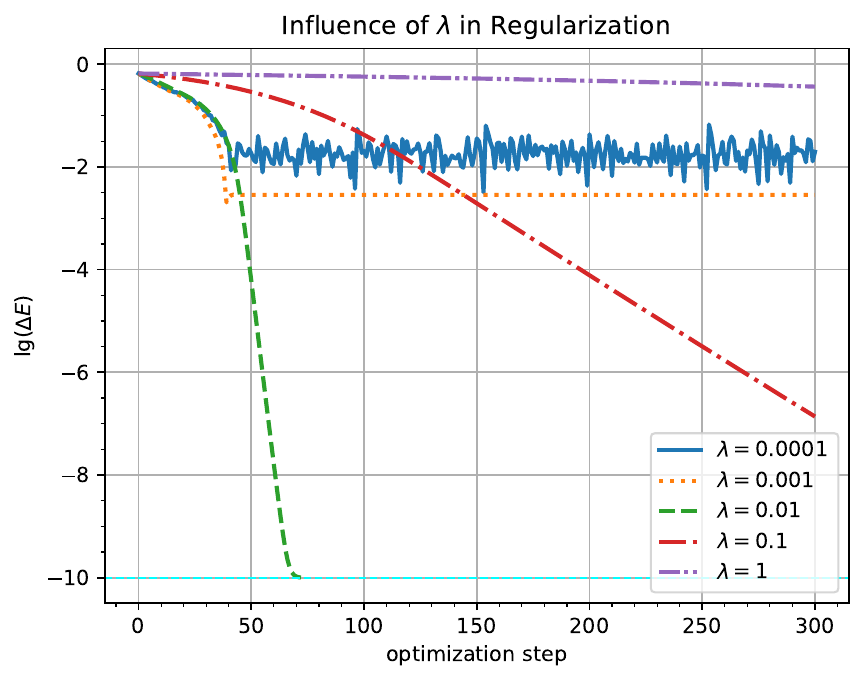}  
    \caption{The influence of $\lambda$ in regularization for H-QNG. The experiment records the energy distance from the ground state $\Delta E$ of a $4$-qubit H\textsubscript{2} Hamiltonian using different values of the regularization hyperparameter $\lambda$ in H-QNG. The learning curves are plotted until a precision of $10^{-10}$ (cyan dashed line). A large value of $\lambda$ leads to slow convergence (purple long-dash double-dot line), while a small $\lambda$ may result in oscillations (blue solid line) or early convergence (orange dotted line).}
    \label{fig-lambda}
\end{figure}

\subsection{Application Examples} \label{NV}

In this section, we evaluate H-QNG against standard QNG and vanilla gradient descent (VG) as baseline methods on two molecular Hamiltonians: the 12-qubit LiH Hamiltonian and the 12-qubit H\textsubscript{6} Hamiltonian. Both Hamiltonians are obtained from the PennyLane Molecules dataset \cite{azad2023pennylaneqchem}. We consider two types of ansatz in our numerical experiments: a chemistry-inspired ansatz provided by the PennyLane Molecules dataset \cite{azad2023pennylaneqchem}, and a four-layer EfficientSU2 ansatz \cite{kandala2017hardware} which is hardware-efficient. The ansatze used in this section are shown in Appendix~\ref{ansatz} in detail.

For the chemistry-inspired ansatz, a solution achieving chemical accuracy is guaranteed to lie within the search space. The initial parameters are generated using two different strategies. First, the local initialization: the parameters are randomly initialized around the optimal solutions provided by PennyLane \cite{azad2023pennylaneqchem} to evaluate the local convergence properties of the different optimization methods. Second, the global initialization: the parameters are randomly sampled from $[-\pi,\pi]$ to evaluate the robustness of the optimizers in locating low-energy regions from global random initialization. The regularization factor $\lambda$ for H-QNG is set to $0.1$ and $0.01$ for the two cases, respectively. We report the achieved energy $E$ for both cases (left panels of \Cref{fig-lih,fig-h6}, and \Cref{fig-global}). For the local initialization, we also report the energy distance $\Delta E$ relative to the ground-state energy on a logarithmic scale to evaluate the convergence rate (right panels of \Cref{fig-lih,fig-h6}). For all numerical experiments in this section, the results are averaged over 10 independent runs.

From \Cref{fig-lih} and \Cref{fig-h6}, we observe that H-QNG achieves competitive performance compared with standard QNG for both Hamiltonians under local initialization, and reaches chemical accuracy in fewer optimization steps. Moreover, since each optimization step of H-QNG requires fewer circuit evaluations than standard QNG, H-QNG achieves chemical accuracy with a lower overall quantum computational cost. In addition, both H-QNG and standard QNG significantly outperform vanilla gradient descent, which fails to achieve chemical accuracy within the given budget of $100$ optimization steps. From \Cref{fig-global}, we observe that, under global random initialization, H-QNG achieves a lower average energy than standard QNG within $100$ optimization steps for both Hamiltonians. Although both H-QNG and QNG can be trapped in local minima and do not converge to the global optimum, H-QNG reaches lower-energy regions on average. Both methods still significantly outperform vanilla gradient descent. Overall, these results indicate that H-QNG exhibits favorable convergence behavior under both local and global initialization. In the local-initialization setting, H-QNG reaches chemical accuracy in fewer optimization steps, while under global random initialization, it tends to reach lower-energy regions than standard QNG within the same optimization budget.

For the EfficientSU2 ansatz, the initial parameters are randomly generated from $[-0.5,0.5]$, and the regularization factor $\lambda$ for H-QNG is set to $0.1$. Similarly, we report the achieved energy $E$ for both Hamiltonians (\Cref{fig-su2}). From the results, we observe that H-QNG achieves a lower average energy than QNG within $100$ optimization steps. Neither optimizer converges to the exact ground-state energy within 100 optimization steps, which can be attributed to two factors: (1) the limited expressivity of the ansatz and (2) the presence of local minima in the optimization landscape. Nevertheless, H-QNG consistently achieves relatively lower energies than QNG in this setting. In addition, both H-QNG and QNG significantly outperform vanilla gradient descent.

In conclusion, these results highlight the potential advantage of H-QNG: it requires the same quantum computational cost as vanilla gradient descent in each step, while achieving performance that is even superior to standard QNG. We speculate that this performance improvement arises from the observation that the physically relevant states occupy only a small subspace of the full Hilbert space \cite{poulin2011quantum, mcardle2019variational}. This may further suggest that, for ground-state optimization problems, focusing on the optimization in the nontrivial manifold spanned by the target Hamiltonian terms could be more effective than considering the entire quantum state space. In addition, another possible reason for H-QNG’s superior performance is that it explicitly incorporates the coefficients of the Hamiltonian terms, which helps the optimizer better weight the step size along different directions for each parameter in the optimization step.

\begin{figure*}[ht]
    \centering
    \begin{minipage}{0.45\textwidth}
        \centering
        \includegraphics[width=\linewidth]{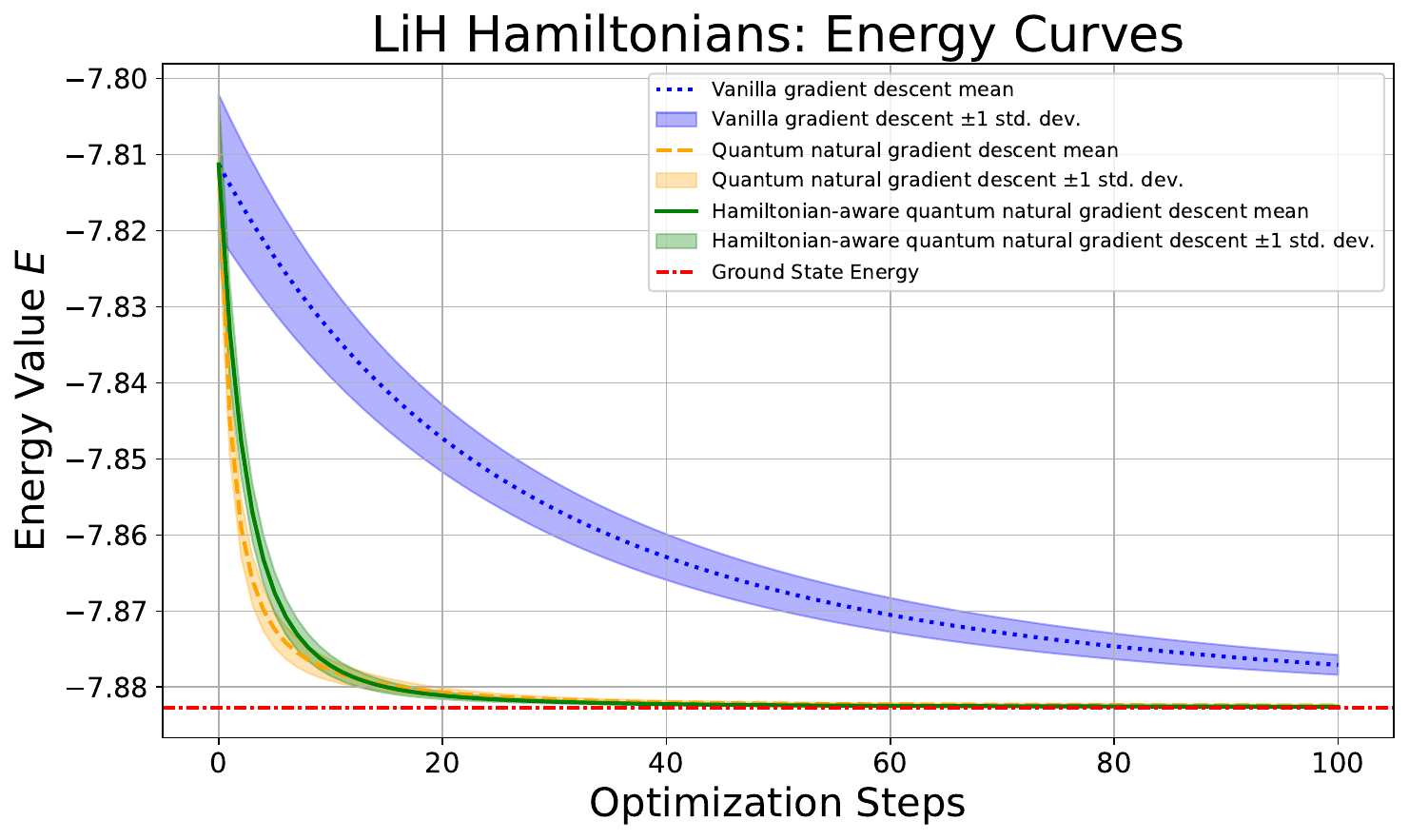}
    \end{minipage}
    \begin{minipage}{0.45\textwidth}
        \centering
        \includegraphics[width=\linewidth]{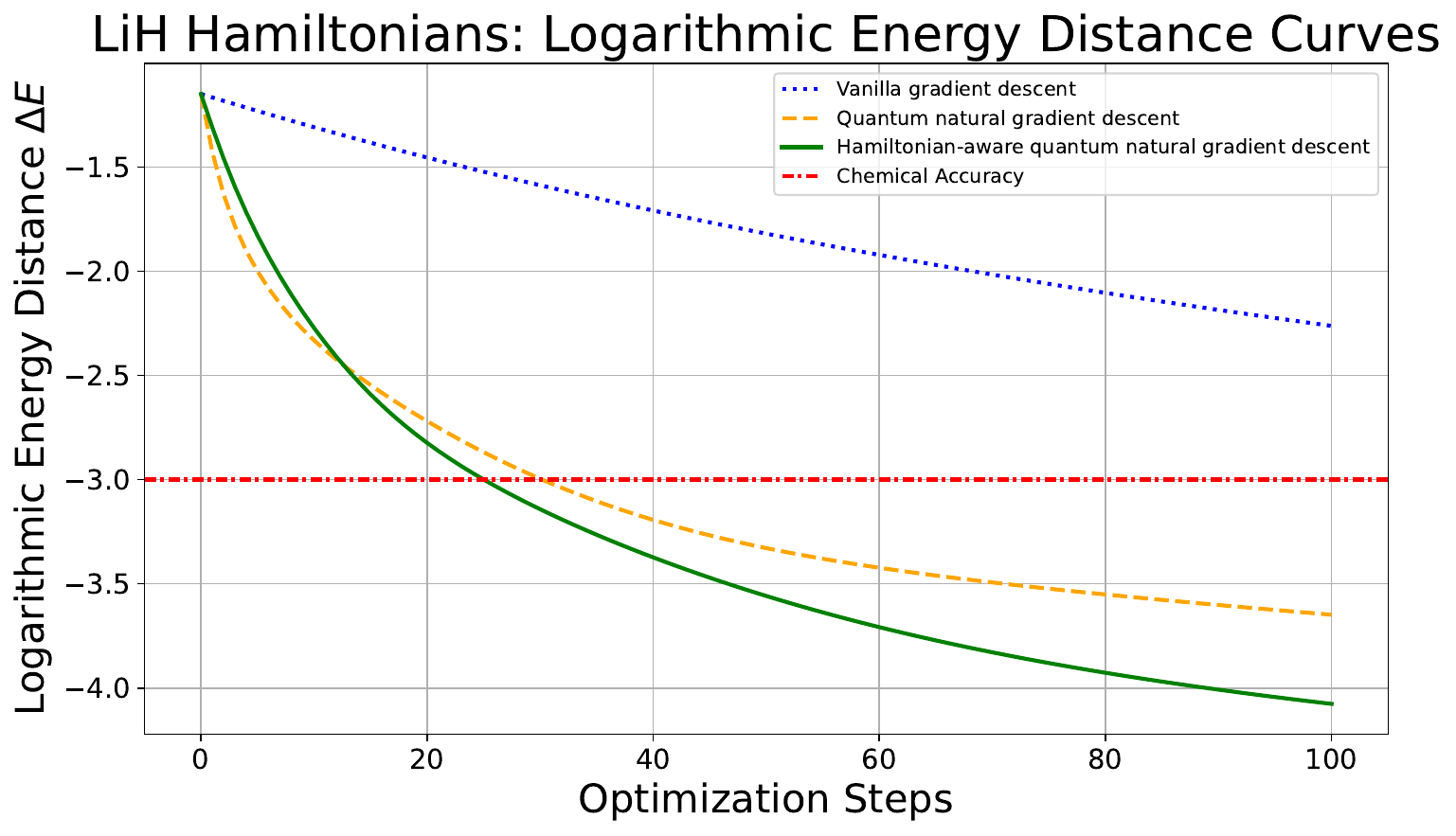}
    \end{minipage}
    \caption{The energy curves (left) and logarithmic energy-distance curves (right) over $10$ independent runs for the LiH Hamiltonian, using local initialization on the chemistry-inspired ansatz.}
    \label{fig-lih}
\end{figure*}

\begin{figure*}[ht]
    \centering
    \begin{minipage}{0.45\textwidth}
        \centering
        \includegraphics[width=\linewidth]{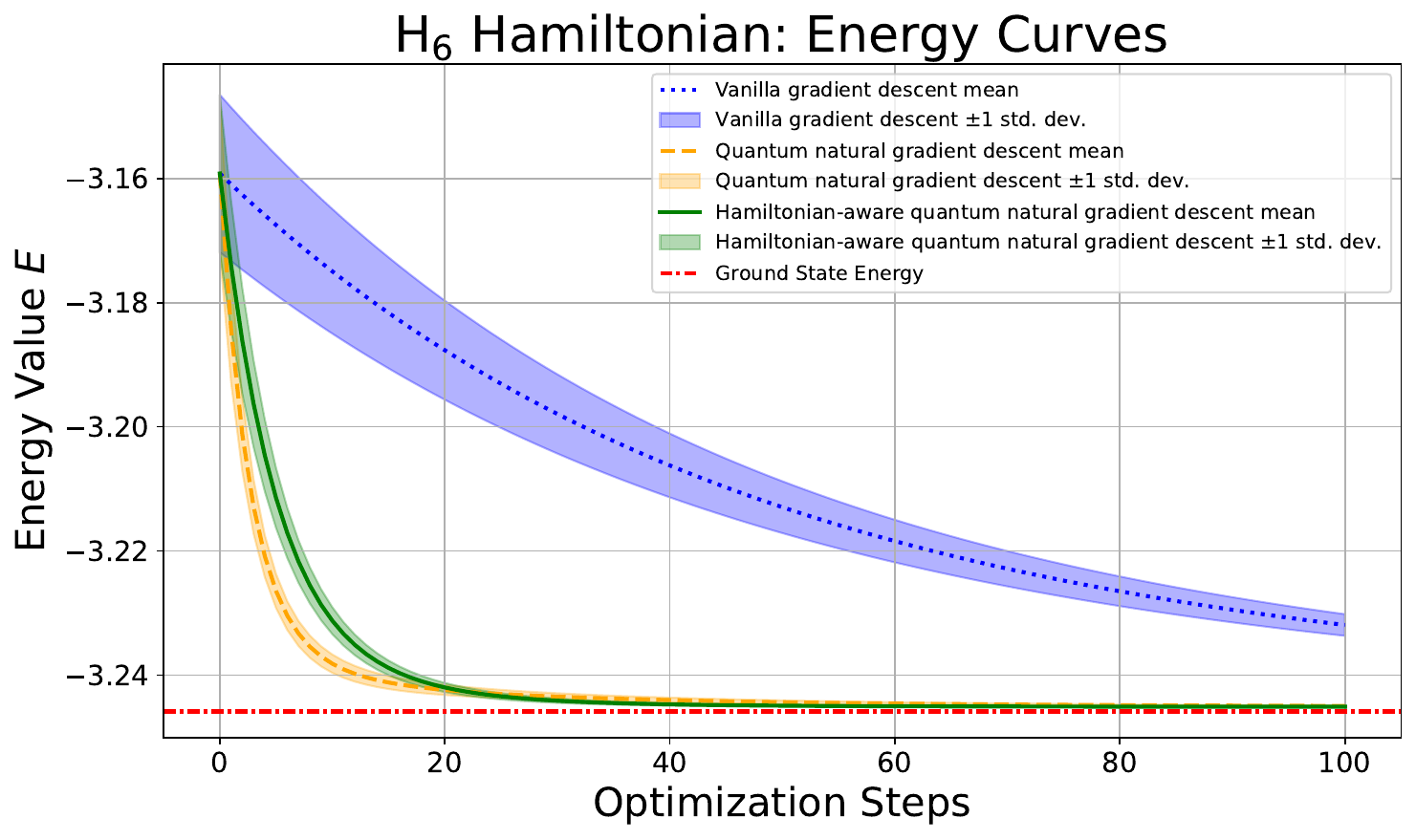}
    \end{minipage}
    \begin{minipage}{0.45\textwidth}
        \centering
        \includegraphics[width=\linewidth]{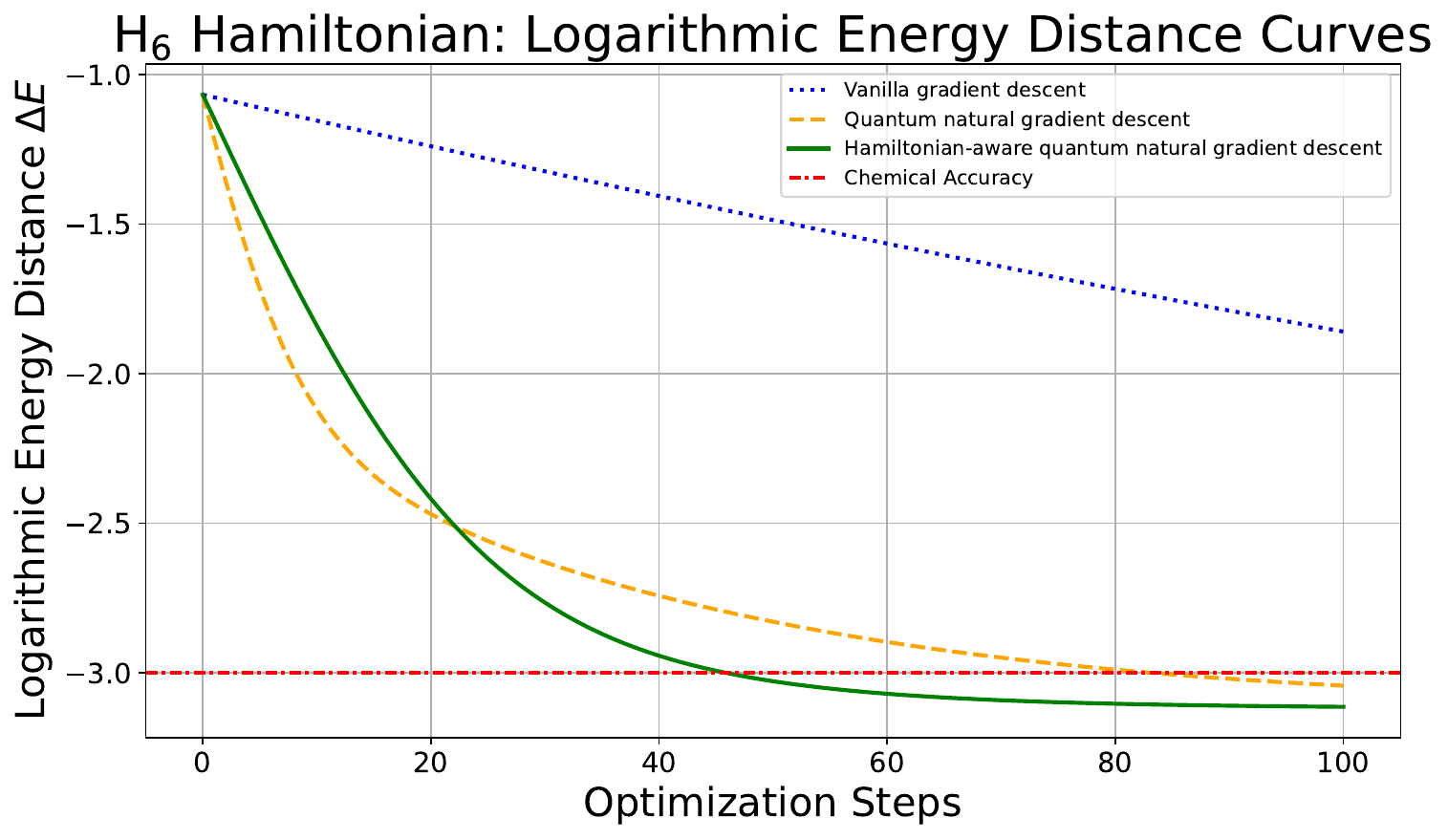}
    \end{minipage}
    \caption{The energy curves (left) and logarithmic energy-distance curves (right) over $10$ independent runs for the H\textsubscript{6} Hamiltonian, using local initialization on the chemistry-inspired ansatz.}
    \label{fig-h6}
\end{figure*}

\begin{figure*}[ht]
    \centering
    \begin{minipage}{0.45\textwidth}
        \centering
        \includegraphics[width=\linewidth]{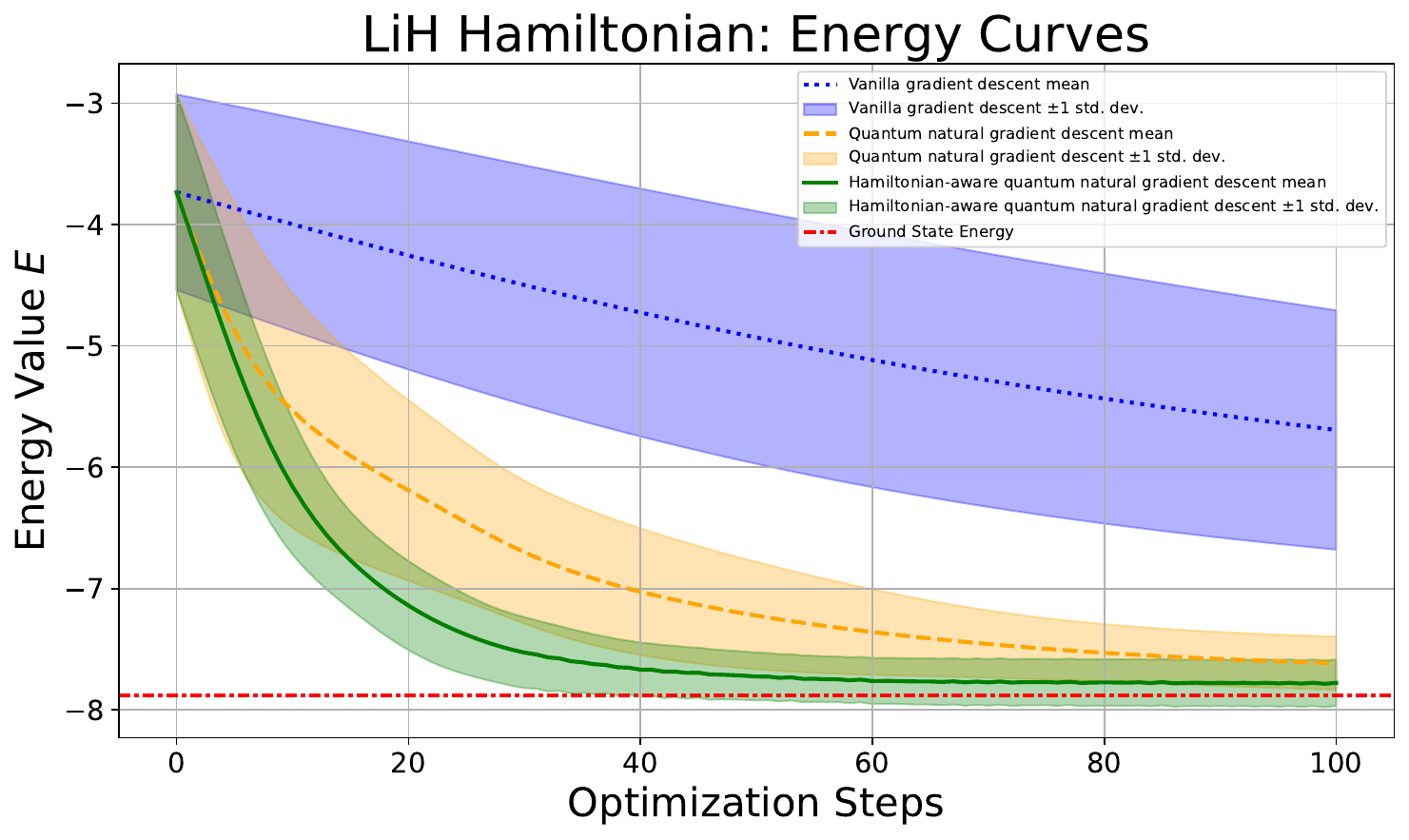}
    \end{minipage}
    \begin{minipage}{0.45\textwidth}
        \centering
        \includegraphics[width=\linewidth]{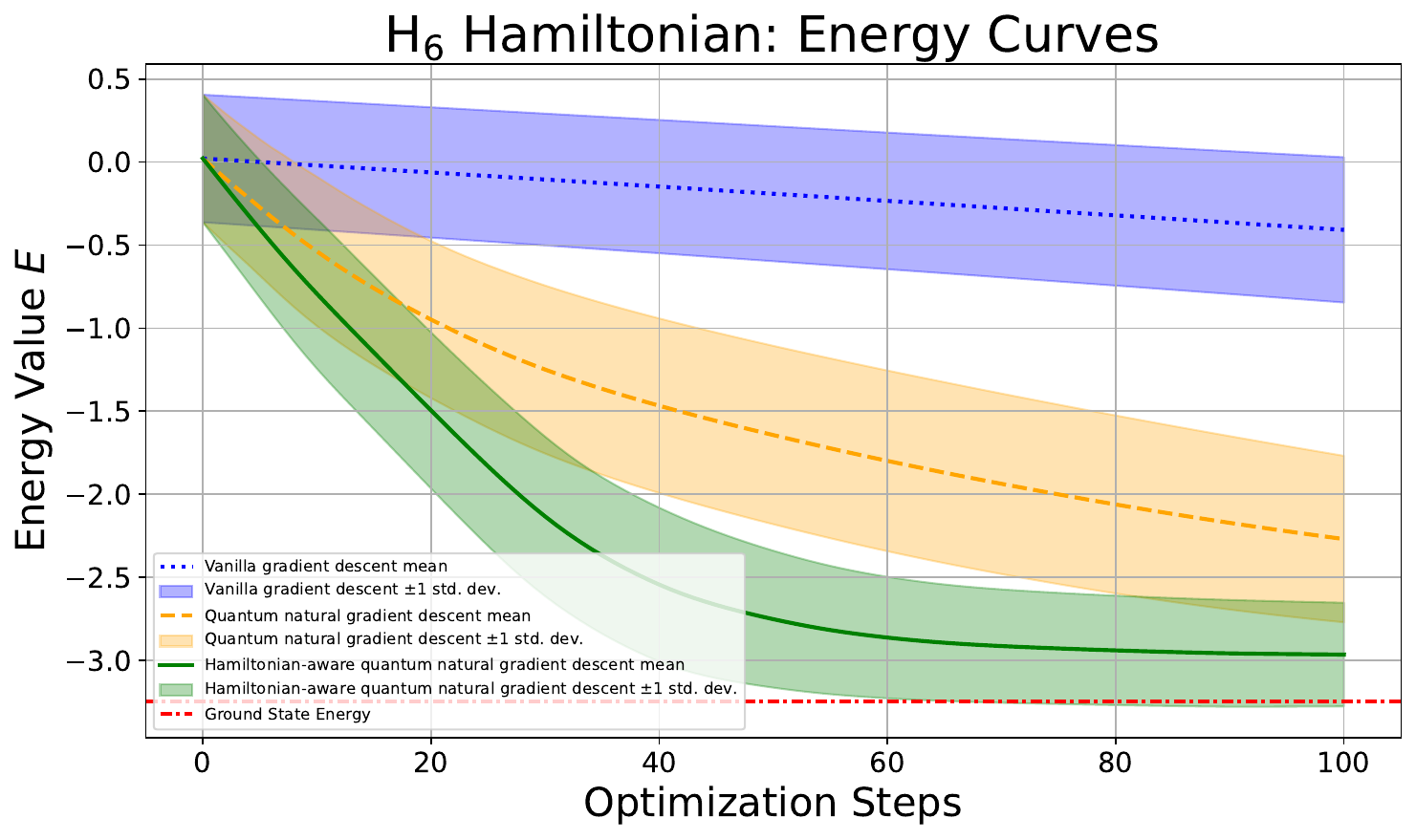}
    \end{minipage}
    \caption{The energy curves over $10$ independent runs for the LiH (left) and H\textsubscript{6} (right) Hamiltonians, using global initialization on the chemistry-inspired ansatz.}
    \label{fig-global}
\end{figure*}

\begin{figure*}[ht]
    \centering
    \begin{minipage}{0.45\textwidth}
        \centering
        \includegraphics[width=\linewidth]{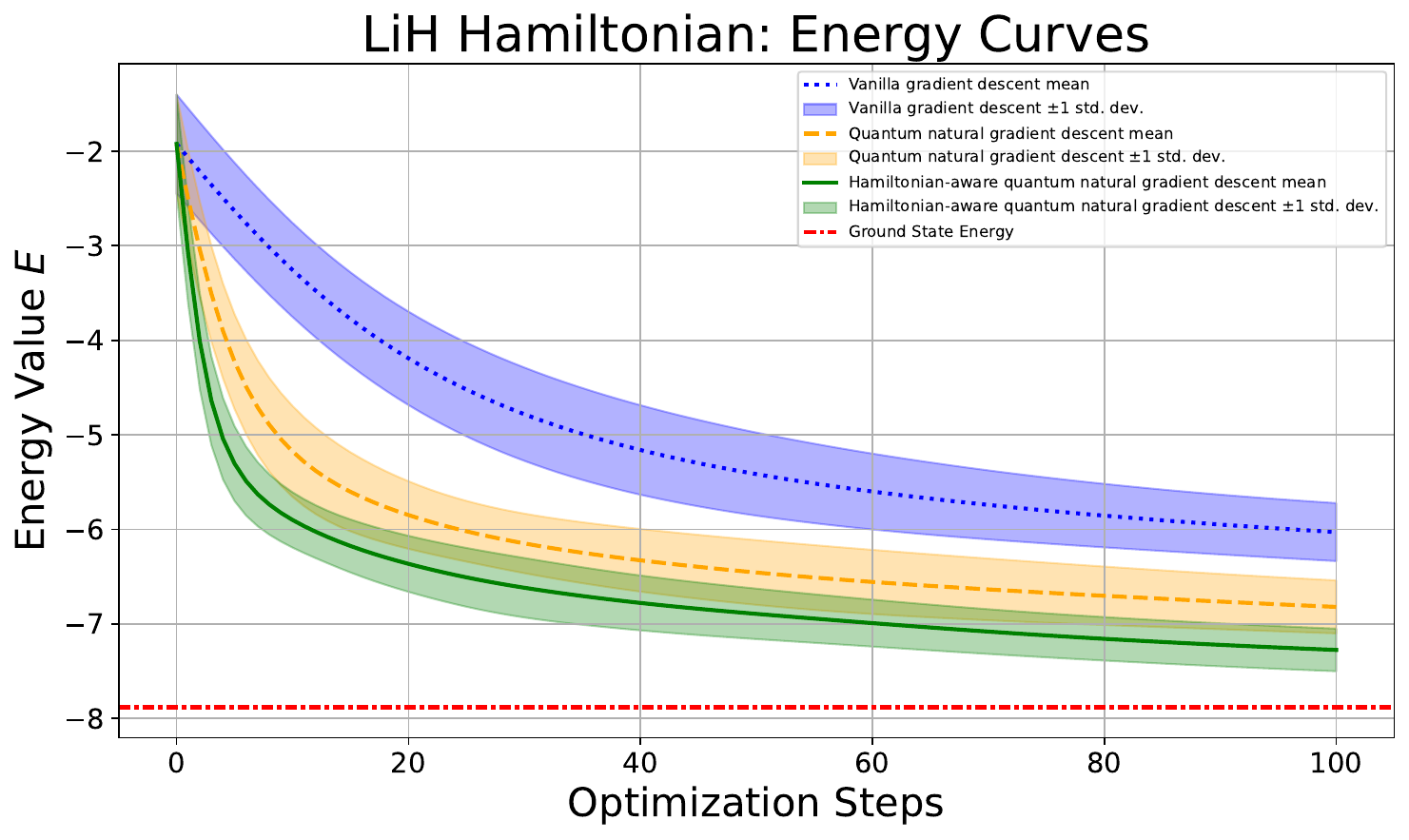}
    \end{minipage}
    \begin{minipage}{0.45\textwidth}
        \centering
        \includegraphics[width=\linewidth]{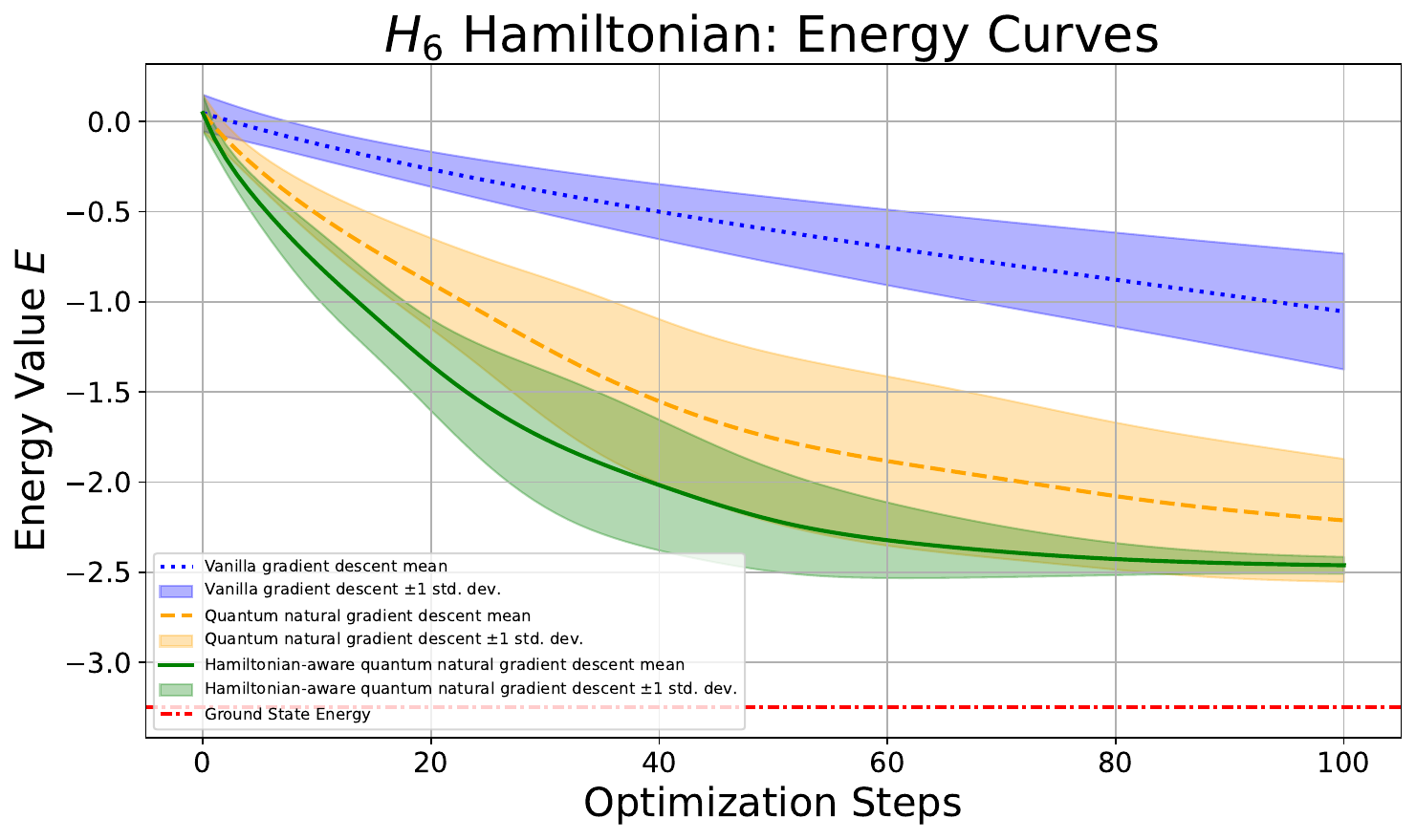}
    \end{minipage}
    \caption{The energy curves over $10$ independent runs for the LiH (left) and H\textsubscript{6} (right) Hamiltonians, using the $4$-layer EfficientSU2 ansatz.}
    \label{fig-su2}
\end{figure*}

\section{Conclusions and Future Work} \label{DC}
In this work, we propose a new optimization method for the Variational Quantum Eigensolver (VQE). Our method is derived from the perspective of a Riemannian pullback metric defined on a manifold that contributes non-trivially to the cost function. Compared to standard QNG, the proposed method takes advantage of the information about the Hamiltonian when computing the metric tensor in each optimization step. The most notable advantage of H-QNG is that it only requires the same quantum computational cost as vanilla gradient descent in each step while achieving performance comparable or even superior to that of standard QNG. Furthermore, we prove that H-QNG preserves key properties of standard QNG, such as reparameterization invariance.

Our experimental results show that H-QNG not only achieves a faster convergence speed but also requires fewer quantum computational resources. In particular, with local initialization, H-QNG reaches chemical accuracy more quickly than standard QNG for three different molecular Hamiltonians: H\textsubscript{2}, LiH, and H\textsubscript{6}. Besides, with wide-range initialization for both the chemistry-inspired ansatz and the EfficientSU2 ansatz, H-QNG also reaches lower-energy regions faster than standard QNG within the optimization step budget. These results indicate the potential advantages of employing H-QNG for the optimization of VQEs. We speculate that this advantage arises from H-QNG’s incorporation of information from the Hamiltonian—specifically, the optimization process is considered directly on the non-trivial manifold spanned by Hamiltonian terms, and the Hamiltonian coefficients are also used to weight the step size along each optimization direction.

In this work, we formulate H-QNG on idealized variational quantum circuits, where the required quantities can be evaluated noiselessly. However, on NISQ devices, circuit noise can also affect the optimization process. Reference \cite{koczor2022quantum} extends standard QNG to the case with circuit noise by employing the Hilbert–Schmidt metric tensor for mixed states instead of the Fubini-Study metric tensor for pure states. For our method, the formulation in \Cref{df-hqng} and the optimization update in \Cref{eq-hqng} can similarly be extended to mixed states and remain well defined, as in \cite{koczor2022quantum}. However, the performance of H-QNG under circuit noise remains to be evaluated. We have also discussed the influence of shot noise in Appendix~\ref{NOISE}.

This work primarily focuses on optimization for VQEs. However, we note that standard QNG can also be applied to other variational quantum algorithms, such as the variational quantum classifier in the field of quantum machine learning. Therefore, H-QNG is also expected to be applicable to problems beyond VQEs, as long as the cost function is defined via the expectation value of a given Hamiltonian. In such cases, our method is expected to provide similar advantages over standard QNG for the same reasons as in VQEs. We consider evaluating the performance of H-QNG on other variational quantum algorithms as future work.

In conclusion, H-QNG provides a promising and efficient optimization method for VQEs. By considering the manifold that contributes non-trivially to the cost function, H-QNG leverages information from the target Hamiltonian, reduces quantum computational costs, and offers a potential research direction for the optimization of variational quantum algorithms. 

\section*{Acknowledgement}
C.S. and H.W. thank Vedran Dunjko for insightful discussions. This project was funded by the European Union (ERC CoG, BeMAIQuantum, 101124342). This project was also co-funded by the Dutch National Growth Fund (NGF), as part of the Quantum Delta NL programme.

\section*{Author Contributions}
This paper was conceived by C.S. and H.W.. C.S. and H.W. formulated the theoretical part. H.W. helped C.S. for the design of numerical experiments. C.S. conducted numerical experiments and analyzed the results with assistance from  H.W. Both authors reviewed the paper.

\section*{Competing Interests}
The authors declare no competing interests.

\clearpage
\appendix
\renewcommand{\thefigure}{A\arabic{figure}}
\setcounter{figure}{0}
\section*{Appendix}
\section{Comparison with OP-VQITE} \label{COPE}
Operator-Projected Variational Quantum Imaginary Time Evolution (OP-VQITE) \cite{anuar2024operator} improves the standard Variational Quantum Imaginary Time Evolution (VQITE) \cite{mcardle2019variational} by operator projection. For a given Hamiltonian $H$, OP-VQITE formulates the metric tensor $G$ and gradient term $b$ as:
\begin{equation}
    G_{ij} = \sum_{O\in S_{\tau}} \tr(\partial_i \rho_{\theta}O)\tr(\partial_j \rho_{\theta}O)
    \label{eq-opmetric}
\end{equation}
\begin{equation}
    b = \sum_{O\in S_{\tau}} V(\rho_{\theta},O)\nabla_{\theta}\tr(\rho_{\theta}O)
    \label{eq-opgradient}
\end{equation}

where $V(\rho_{\theta},O)=\tr(\rho_{\theta}\{O,H\})-2\tr(\rho_{\theta}H)\tr(\rho_{\theta}O)$ and $S_{\tau}$ is an operator set that spans the Hamiltonian $H$. The corresponding parameter update rule is then formulated as:
\begin{equation}
    \theta^{(k+1)} = \theta^{(k)} - \eta G^{-1}b
    \label{eq-opupdate}
\end{equation}

There can be plenty of choices for the operator set $S_{\tau}$. In general, H-QNG and OP-VQITE differ both in the metric tensor and in the gradient term.  In a special case, for a Hamiltonian $H=\sum_{r=1}^v a_r P_r$, if the operator set is chosen as $S_{\tau}=\{a_r P_r \mid 1\leq r\leq v\}$, 
then the metric tensor $G$ defined in \Cref{eq-opmetric} for OP-VQITE becomes equivalent to the metric tensor $T$ defined in \Cref{df-hqng} up to a constant prefactor. However, this equivalence of metric tensors still does not render H-QNG equivalent to OP-VQITE,  because their gradient terms differ: the vector $b$ defined in \Cref{eq-opgradient} for OP-VQITE is not the gradient of the cost function. Since OP-VQITE requires estimating the extra expectation value terms in $V(\rho_{\theta},O)$, whereas H-QNG only requires the same estimates as vanilla gradient descent, H-QNG generally has lower circuit evaluation complexity than OP-VQITE. This is a potential advantage of H-QNG over OP-VQITE.

\section{Impact of Shot Noise} \label{NOISE}

In this work, we mainly consider the ideal case where the circuit is noiseless and the required gradients and metric tensors can be precisely tracked. However, due to NISQ devices, noise can be a great matter for the optimization performance. In our case, a natural follow-up question related to noise is whether the impact of shot noise could become a serious issue for H-QNG, due to the fewer shots used in each step of H-QNG, where both the gradient and the metric tensor are estimated using the same measurement samples. In this section, we conduct additional experiments to show that the influence of shot noise on H-QNG is not worse than that on standard QNG for reaching chemical accuracy. Moreover, H-QNG still displays its fast convergence advantage in reaching the desired accuracy.

The experimental setup is the same as that in \Cref{SAS} for the $\operatorname{H}_2$ Hamiltonian, except that we now include shot noise. Each element used to compute the gradient and the metric tensor is estimated using $1024$ measurement shots instead of being tracked precisely. Note that, because standard QNG requires estimating more elements than H-QNG, each optimization step of QNG consequently involves more measurement shots in total. The optimization step budget is increased to $800$ for standard QNG to better display the convergence behavior of it. The learning curves are shown in \Cref{fig-noise}. The cost function values for plotting learning curves are still obtained exactly, in order to better analyze the impact of shot noise on the optimization process itself. 

\begin{figure}[h]
    \centering
    \includegraphics[width=0.45\textwidth]{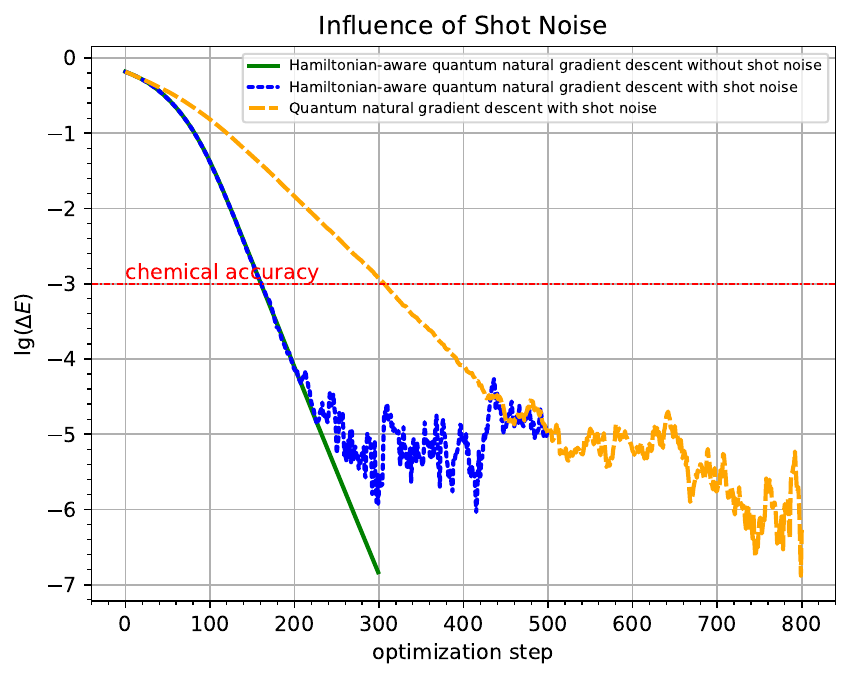}  
    \caption{The learning curves of H-QNG (blue dotted line) and QNG (orange dashed line) under shot noise, with $1024$ measurement shots used for each required quantity. H-QNG still reaches chemical accuracy faster than standard QNG in the presence of shot noise. H-QNG behaves similarly to the noiseless case in the early stages but eventually encounters a precision floor due to shot noise. A similar precision floor is also observed for standard QNG.}
    \label{fig-noise}
\end{figure}

As shown in the figure, the learning curve of H-QNG with shot noise (blue dotted line) still reaches chemical accuracy faster than standard QNG with shot noise (orange dashed line). Moreover, H-QNG with shot noise performs almost identically to the noiseless case (green solid line) in the early stages, where the energy gap to the target state is large, but eventually reaches an accuracy floor due to shot noise. A similar accuracy floor also appears for standard QNG, even though it uses more measurement shots in total. These results indicate that shot noise does affect the optimization process; however, with a sufficient number of measurement shots, both methods are only mildly affected before achieving the desired accuracy. In addition, H-QNG continues to display an advantage in convergence speed compared to QNG under shot noise.

\section{Proofs in Section \ref{Pre}} \label{cons}

Here we provide a detailed derivation from the constrained optimization problem to the update formulas for vanilla gradient descent and QNG. For vanilla gradient descent, using the method of Lagrange multipliers, the constrained optimization problem given in \Cref{vanilla_problem} can be formulated as:
\begin{equation}
    \delta^{*}=\argmin_{\delta} f(\theta + \delta) + \lambda (\Vert \delta \Vert_2^2-\epsilon)
    \label{eq-vgd1}
\end{equation}

By applying the first-order Taylor expansion to $f(\theta+\delta)$, we obtain:
\begin{equation}
    \delta^{*}\approx\argmin_{\delta} f(\theta)+\nabla f(\theta)^T\delta + \lambda (\Vert \delta \Vert_2^2-\epsilon)
    \label{eq-vgd2}
\end{equation}

Here, we use the stationarity condition, which requires that the derivative of the Lagrangian with respect to $\delta$ be zero:
\begin{align}
    0 &= \nabla f(\theta) +2\lambda \delta^* \notag \\
    \delta^* & = -\frac{1}{2\lambda} \nabla f(\theta)
    \label{eq-vgd3}
\end{align}

By considering $\tfrac{1}{2\lambda}$ as the learning rate $\eta$, we derive the same update rule of vanilla gradient descent in \Cref{eq-vanilla}.

For QNG, the derivation proceeds in a similar way. The constrained optimization problem in \Cref{qng_problem} can be formulated as:
\begin{equation}
    \delta^{*}=\argmin_{\delta} f(\theta + \delta) + \lambda (D_F(\rho_{\theta},\rho_{\theta+\delta})-\epsilon)
    \label{eq-qngd1}
\end{equation}

Note that for pure states, we have $D_F(\rho_{\theta},\rho_{\theta+\delta})=\frac{1}{2}\Vert\rho_{\theta+\delta}- \rho_{\theta}\Vert_2^2$. By applying first-order Taylor expansion to $\rho_{\theta+\delta}$, we have $\rho_{\theta+\delta}\approx\rho_{\theta}+\sum_i\partial_i\rho_{\theta}\delta_i$. Hence, we have $D_F(\rho_{\theta},\rho_{\theta+\delta}) \approx \frac{1}{2}\Vert\sum_i\partial_i\rho_{\theta}\delta_i\Vert_2^2=\frac{1}{2}\sum_{ij}\tr(\partial_i\rho_{\theta}\partial_j\rho_{\theta})\delta_i\delta_j=\delta^TA\delta$, where the matrix $A_{ij}=\frac{1}{2}\tr(\partial_i\rho_{\theta}\partial_j\rho_{\theta})$ is equivalent to the Fubini-Study metric tensor when the state $\rho_{\theta}$ is pure \cite{koczor2022quantum, shi2025weighted}. Substituting this expression into \Cref{eq-qngd1} and applying the first-order Taylor expansion to $f(\theta+\delta)$, we obtain:
\begin{equation}
    \delta^{*}\approx\argmin_{\delta} f(\theta)+\nabla f(\theta)^T\delta + \lambda (\delta^TA\delta-\epsilon)
    \label{eq-qngd2}
\end{equation}

Similarly, by applying the stationary condition, we obtain:
\begin{align}
    0 &= \nabla f(\theta) +2\lambda A\delta^* \notag \\
    \delta^* & = -\frac{1}{2\lambda} A^{-1}\nabla f(\theta)
    \label{eq-qngd3}
\end{align}
By considering $\frac{1}{2\lambda}$ as the learning rate $\eta$, we derive the same update rule of QNG in \Cref{df-qng}.

\section{Proofs in Section \ref{MF}} \label{pf-1}
Here we provide a proof of \Cref{pp-qng}. The function $h:\mathcal{S}\rightarrow \mathcal{M}$ in \Cref{pp-qng} maps a parameter $\theta \in R^m$ to a vector $[\tr(\rho_{\theta}P_1),\cdots \tr(\rho_{\theta}P_{4^n})]\in R^{4^{n}}$, and the corresponding function $g:\mathcal{M}\rightarrow \mathbb{R}$ is a linear combination of these $4^n$ elements such that $g\circ h(\theta)=f(\theta)=\tr(\rho_\theta H)$. The main objective is to show that the pullback metric $G$ defined in \Cref{pp-qng} is equal to the Fubini-Study metric tensor $A$ up to a constant factor $2^{n+1}$: $G=2^{n+1}A$. According to \Cref{df-pullback}, we have:
\begin{align}
    G_{ij} &= \langle[\tr(\partial_i\rho P_1),\cdots,\tr(\partial_i\rho P_{4^n})],[\tr(\partial_j\rho P_1),\cdots,\tr(\partial_j
    \rho P_{4^n})]\rangle \notag \\
    &=\sum_{r=1}^{4^n} \tr(\partial_i\rho P_r)\tr(\partial_j \rho P_r) \notag \\
    &= 2^n\tr((\partial_i \rho \otimes\partial_j \rho)S) \notag \\
    &= 2^n\tr(\partial_i\rho \partial_j\rho) \notag \\
    &= 2^{n+1} A_{ij}
    \label{g2a}
\end{align}

The second equality follows from the fact that $G_{\mathcal{M}}$ is the Euclidean Riemannian metric. The third equality results from the Pauli decomposition of the swap operator $S$. The fourth equality applies the swap trick \cite{mele2024introduction}. The final equality uses the relation between the Hilbert–Schmidt metric tensor $T$ and the Fubini-Study metric tensor $A$, namely $T_{ij}=2A_{ij}$ \cite{koczor2022quantum, shi2025weighted}.

Following the same derivation approach as in \Cref{g2a}, we can obtain the expression for the pullback metric $G$ of H-QNG defined in \Cref{df-hqng}, where $G_{ij} = \sum_r^v a_r^2\tr(\partial_i\rho_{\theta}P_r)\tr(\partial_j\rho_{\theta}P_r)$.

\section{Proofs in Section \ref{RI}} \label{pf-2}
Here we provide a brief proof of \Cref{th-inva}. \Cref{pp-inva} can be proved in a similar way. To prove \Cref{th-inva}, we introduce the following \Cref{lm-ode} for ODEs. 

\begin{lemma}
    For two ODEs $\dot{\theta}=A(\theta)$ and $\dot{\psi}=B(\psi)$ with smooth vector fields and a diffeomorphism reparameterization $h$ between $\theta$ and $\psi$ where $\theta(0)=h(\psi(0))$, if for any $\theta=h(\psi)$ it holds that $A(\theta)=J_h(\psi)B(\psi)$ where $J_h(\psi)$ is the Jacobian of $h$ at $\psi$, then $\theta(t)=h(\psi(t))$ for all $t\geq0$ on their common interval of existence.
    \label{lm-ode}
\end{lemma}
\begin{proof}
Define $\bar{\theta}(t)=h(\psi(t))$. By the chain rule and the assumption of the lemma, we have
\[
\dot{\bar{\theta}}(t)
=J_h(\psi(t))\dot{\psi}(t)
=J_h(\psi(t))B(\psi(t))
=A(h(\psi(t)))
=A(\bar{\theta}(t)).
\]
Moreover, $\bar{\theta}(0)=h(\psi(0))=\theta(0)$.
Thus, $\bar{\theta}(t)$ and $\theta(t)$ satisfy the same ODE with the same initial condition. By uniqueness of the solution, we obtain
$\theta(t)=\bar{\theta}(t)=h(\psi(t))$
on their common interval of existence.
\end{proof}

According to \Cref{lm-ode}, to prove \Cref{th-inva}, we only need to prove that, for any $\theta=h(\psi)$, it holds that $T_{\theta}^{-1}\nabla_{\theta} f = J_h(\psi)T_{\psi}^{-1}\nabla_{\psi}f$. Here we introduce \Cref{lm-jac}.

\begin{lemma}
    The cost function $f$ and metric tensor $T$ of H-QNG are defined by \Cref{df-hqng}. For a diffeomorphism reparameterization function $h$ between $\theta$ and $\psi$, when $\theta=h(\psi)$, then we have $\nabla_\psi f=J_h^{T}(\psi)\nabla_{\theta} f$ and $T_{\psi}=J_h^{T}(\psi)T_{\theta}J_h(\psi)$.
    \label{lm-jac}
\end{lemma}
\begin{proof}
    Suppose $\theta,\psi \in \mathbb{R}^m$. We use Einstein summation notation \cite{barr1991einstein} in the derivation. For the first equation to prove, viewing the gradient as an $m\times1$ matrix for notation convenience, we have:
    \begin{align*}
        [\nabla_{\psi}f]_{i1} &= \frac{\partial f}{\partial \psi_i} \\
        &=\frac{\partial f}{\partial \theta_j}\frac{\partial \theta_j}{\partial \psi_i} \\
        &= [J_h]^{ji}\frac{\partial f}{\partial \theta_j} \\
        &= [J_h^T]^{ij} [\nabla_{\theta}f]_{j1}
    \end{align*}
    The above equation is equivalent to $\nabla_\psi f=J_h^{T}(\psi)\nabla_{\theta} f$. For the second equation to prove, we have:
    \begin{align*}
        [T_{\psi}]_{ij} &= \frac{1}{2\sqrt{\sum_{r=1}^va_r^2}}\sum_{r=1}^v a_r^2\tr(\frac{\partial \rho}{\partial \psi_i}P_r)\tr(\frac{\partial \rho}{\partial\psi_j}P_r) \\
        &=\frac{1}{2\sqrt{\sum_{r=1}^va_r^2}}\sum_{r=1}^v a_r^2\tr(\frac{\partial \rho}{\partial \theta_k}\frac{\partial \theta_k}{\partial \psi_i}P_r)\tr(\frac{\partial \rho}{\partial\theta_l}\frac{\partial \theta_l}{\partial \psi_j}P_r) \\
        &= \frac{1}{2\sqrt{\sum_{r=1}^v a_r^2}}
        \sum_{r=1}^v a_r^2
        \tr([J_h]^{ki}\frac{\partial\rho}{\partial\theta_k}P_r)
        \tr([J_h]^{lj}\frac{\partial\rho}{\partial\theta_l}P_r) \\
        &= [J_h^T]^{ik}[T_{\theta}]_{kl}[J_h]^{lj}
    \end{align*}
    The above equation is equivalent to $T_{\psi}=J_h^{T}(\psi)T_{\theta}J_h(\psi)$.
\end{proof}

Now we can prove $T_{\theta}^{-1}\nabla_{\theta}f=J_h(\psi)T_{\psi}^{-1}\nabla_{\psi}f$ from the LHS:

\begin{align}
    T_{\theta}^{-1}\nabla_{\theta} f &= (J_h^{-T}(\psi)T_{\psi}J_h^{-1}(\psi))^{-1}J_h^{-T}(\psi)\nabla_{\psi}f\\
    &=J_h(\psi)T_{\psi}^{-1}\nabla_{\psi}f
\end{align}

The first equation follows from \Cref{lm-jac}. Now, with this equation and \Cref{lm-ode}, the parameterization invariance of H-QNG is proven. The parameterization invariance of standard QNG can be proven in a similar way.

\section{Proofs in Section \ref{SMT}} \label{pf-3}
We begin by introducing the following lemma, which will be used in the proof presented in this section.
\begin{lemma}
    Let $H(\theta)$ be a real analytic function on a connected open domain $U \subseteq \mathbb{R}^m$. If $H$ is not identically zero, then its zero set $Z(H) \coloneqq \{\theta\in U:H(\theta)=0\}$ has a zero measure, namely $\operatorname{mes}_mZ(H)=0$.
    \label{lm-zero}
\end{lemma}

The proof and discussion of \Cref{lm-zero} can be found in \cite{mityagin2015zero}. Therefore, to establish \Cref{pp-sinqng} and \Cref{th-sinhqng}, it suffices to show that $\det(A(\theta))$ and $\det(T(\theta))$ are real analytic functions. According to \cite{koczor2022quantum, koczor2022quantum2}, each expectation value $h(\theta)=\tr(\rho_{\theta}P)$ with a Pauli string $P$ forms a trigonometric series. For a circuit of finite depth and width, this series is a finite linear combination of trigonometric functions, and hence is analytic. The determinant of the metric tensor with elements $T_{ij}=\frac{1}{2\sqrt{\sum_{r=1}^va_r^2}}\sum_{r=1}^{v}a_r^2\tr(\partial_i\rho_{\theta}P_r)\tr(\partial_j\rho_{\theta}P_r)$, is also a finite combination of such analytic functions, and is therefore itself analytic.

Therefore, according to \Cref{lm-zero}, the zero set $\det(T(\theta))=0$ must have measure zero in the parameter space (corresponding to the first case in \Cref{th-sinhqng}), unless the determinant is constantly zero (corresponding to the second case in \Cref{th-sinhqng}). The proof of \Cref{pp-sinqng} follows by the same argument.

\section{Ansatze used in Numerics}
\label{ansatz}

In this section, we explicitly show the ansatze used in our numerics throughout the paper. 

The ansatz used for H\textsubscript{2} in \Cref{fig-scaling,fig-lambda,fig-noise} is shown in \Cref{fig-twoansatz} (left). This ansatz contains $4$ $R_x$ gates and $1$ double-excitation gate $U^{(2)}$. The search space of the ansatz is guaranteed to contain a state that achieves chemical accuracy.

The ansatz used for evaluating reparameterization invariance in \Cref{fig-ri} is shown in \Cref{fig-twoansatz} (right). This ansatz contains $2$ $R_y$ gates and $1$ CNOT gate. The exact ground state of the given Hamiltonian lies within the search space of this ansatz.

\begin{figure*}[ht]
    \centering
    \begin{minipage}{0.45\textwidth}
        \centering
        \includegraphics[width=\linewidth]{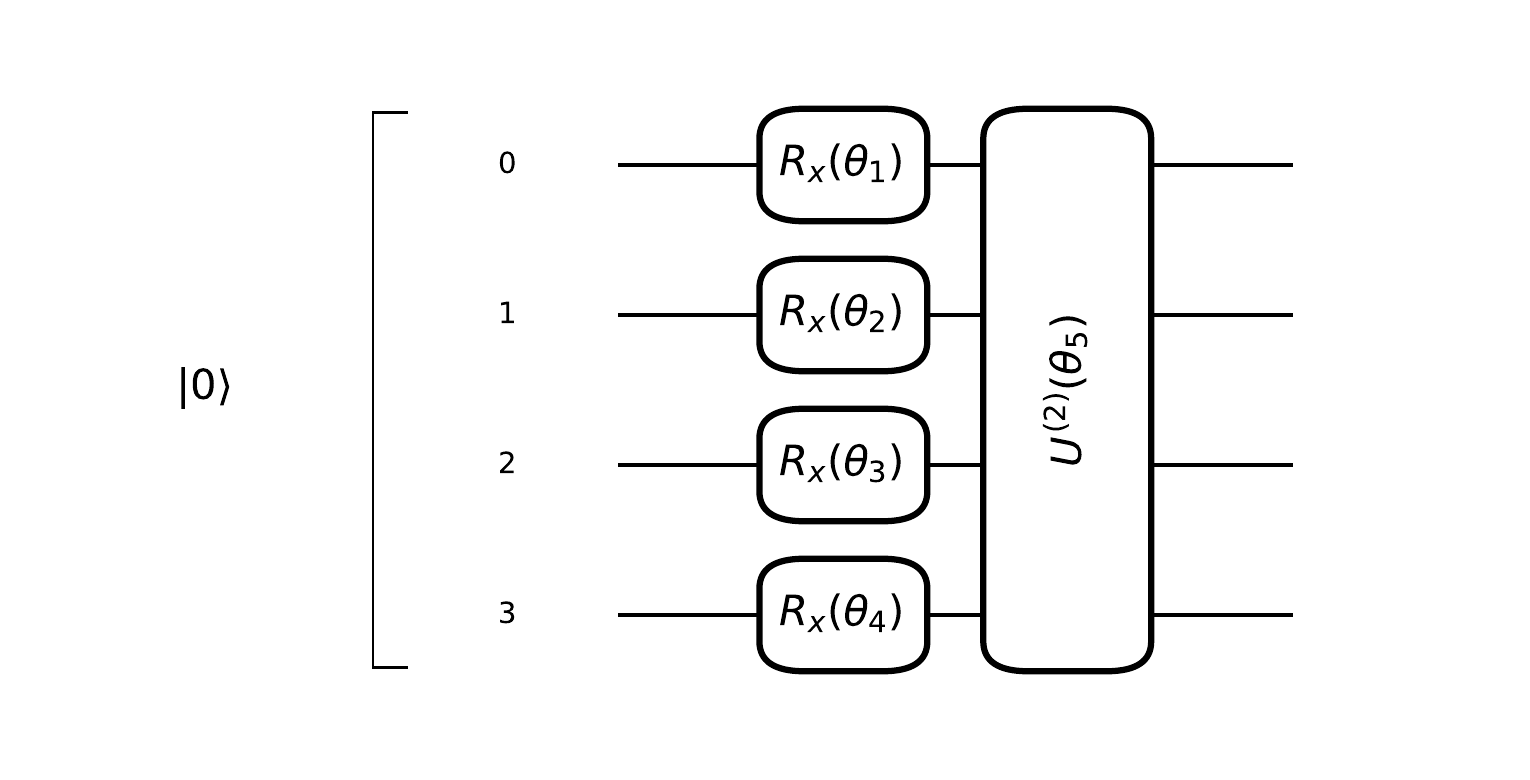}
    \end{minipage}
    \begin{minipage}{0.45\textwidth}
        \centering
        \includegraphics[width=\linewidth]{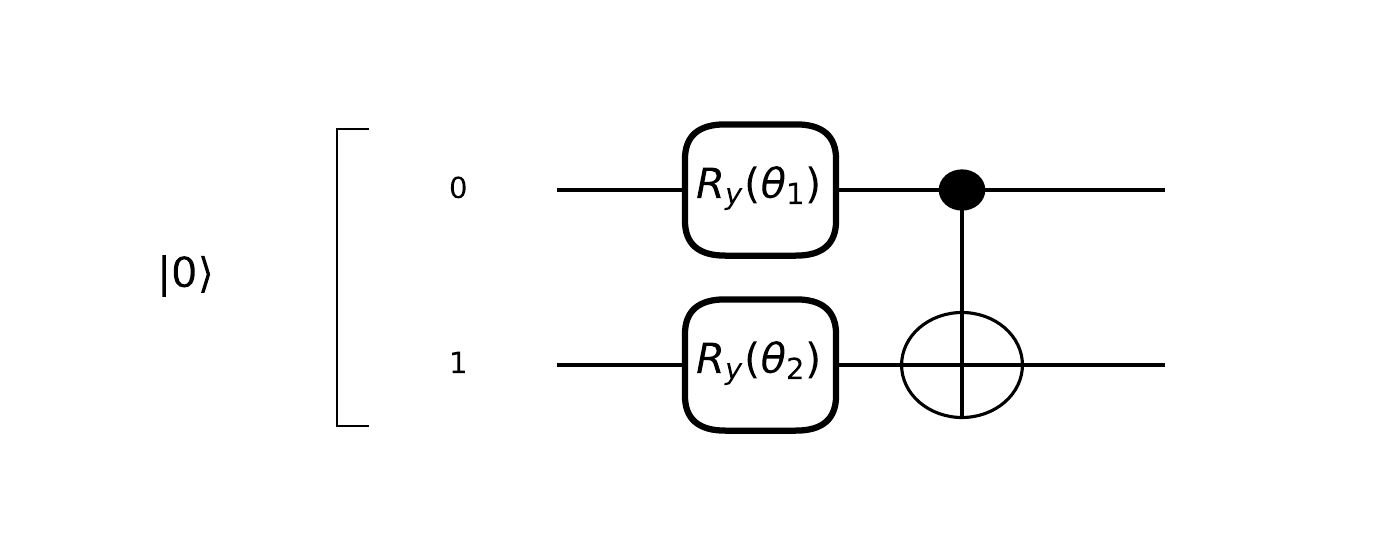}
    \end{minipage}
    \caption{Ansatz used for H\textsubscript{2} in \Cref{fig-scaling,fig-lambda,fig-noise} (left) and for evaluating reparameterization invariance in \Cref{fig-ri} (right).}
    \label{fig-twoansatz}
\end{figure*}

The chemistry-inspired ansatz from the PennyLane Molecules dataset \cite{azad2023pennylaneqchem}, used for the LiH Hamiltonian in \Cref{fig-lih} and in the left panel of \Cref{fig-global}, is shown in \Cref{fig-ansatzlih}. This ansatz starts from the corresponding Hartree--Fock state and contains $26$ DoubleExcitation gates $U^{(2)}$ and $8$ SingleExcitation gates $U^{(1)}$. The search space of this ansatz is guaranteed to contain a state that achieves chemical accuracy.

\begin{figure}[h]
    \centering
    \includegraphics[width=1.0\textwidth]{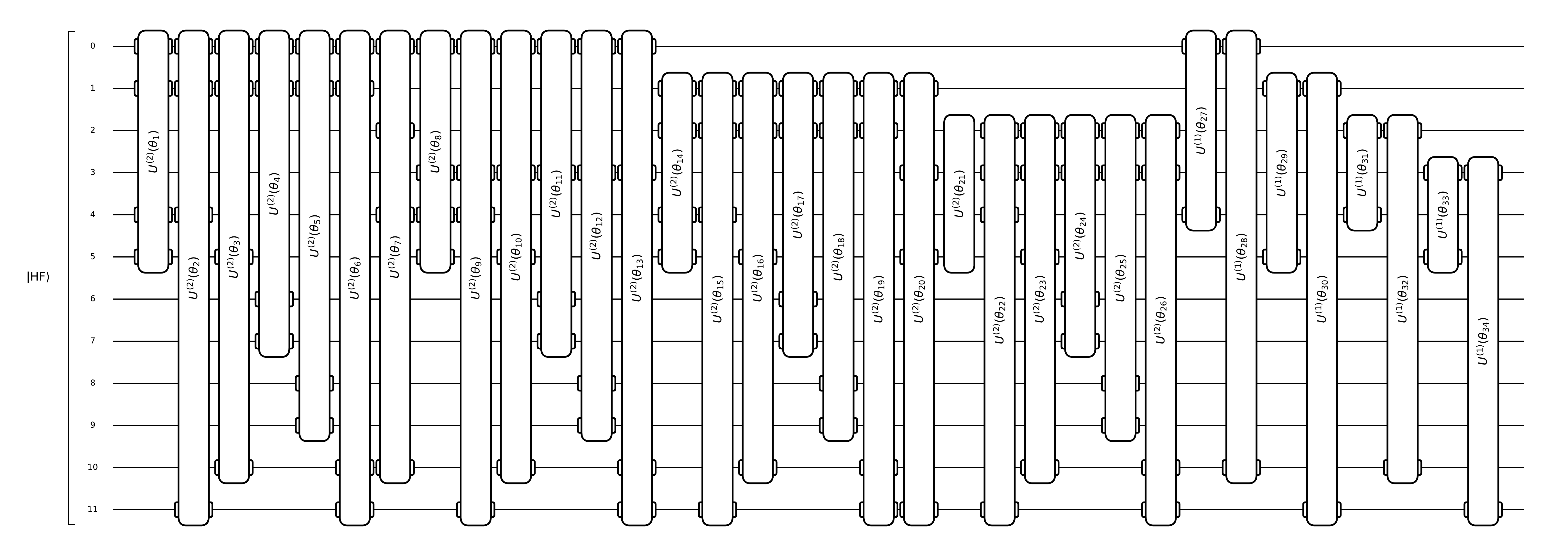}  
    \caption{Chemistry-inspired ansatz used for LiH in \Cref{fig-lih} and in the left panel of \Cref{fig-global}.}
    \label{fig-ansatzlih}
\end{figure}

The chemistry-inspired ansatz from the PennyLane Molecules dataset \cite{azad2023pennylaneqchem}, used for the H\textsubscript{6} Hamiltonian in \Cref{fig-h6} and in the right panel of \Cref{fig-global}, is shown in \Cref{fig-ansatzh6}. This ansatz starts from the corresponding Hartree--Fock state and contains $51$ DoubleExcitation gates $U^{(2)}$ and $8$ SingleExcitation gates $U^{(1)}$. The search space of this ansatz is guaranteed to contain a state that achieves chemical accuracy.

\begin{figure}[h]
    \centering
    \includegraphics[width=1.0\textwidth]{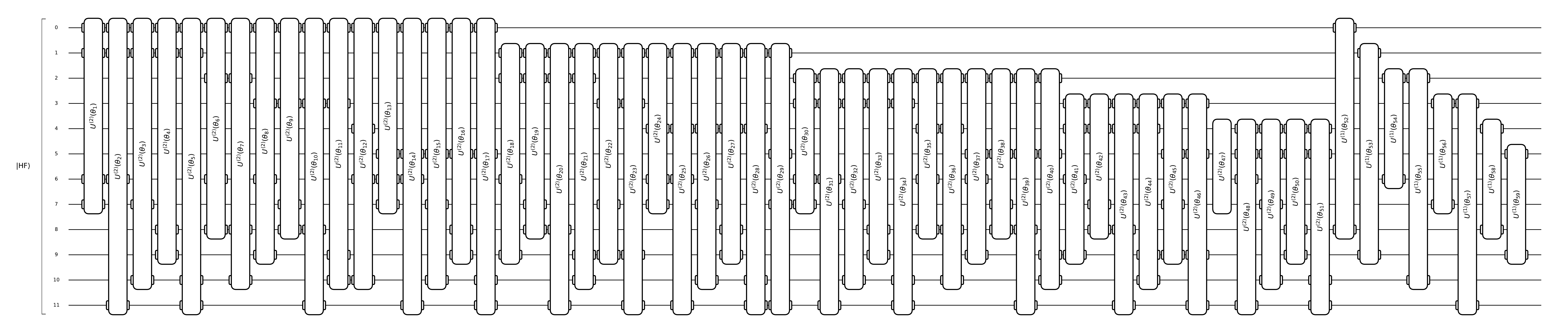}  
    \caption{Chemistry-inspired ansatz used for H\textsubscript{6} in \Cref{fig-h6} and in the right panel of \Cref{fig-global}.}
    \label{fig-ansatzh6}
\end{figure}

The $4$-layer EfficientSU2 ansatz used for the LiH and H\textsubscript{6} Hamiltonians in \Cref{fig-su2} is shown in \Cref{fig-efficientsu2}. The ansatz starts from the corresponding Hartree--Fock state and contains $48$ $R_y$ gates, $48$ $R_z$ gates, and $44$ CNOT gates. This ansatz is considered to be hardware-efficient.

\begin{figure}[h]
    \centering
    \includegraphics[width=1.0\textwidth]{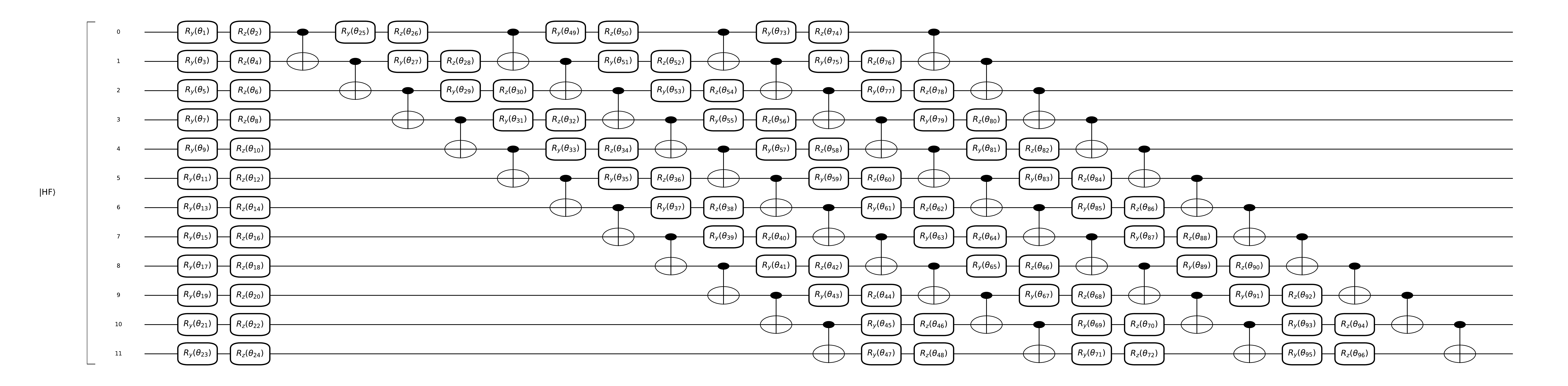}  
    \caption{EfficientSU2 ansatz used for both LiH and H\textsubscript{6} in \Cref{fig-su2}.}
    \label{fig-efficientsu2}
\end{figure}

\bibliographystyle{unsrt}  
\bibliography{references}
\end{document}